\documentclass[12pt]{article}
\usepackage{amssymb,amsmath,amsfonts,amssymb}
\usepackage{graphics,graphicx,color,hhline}
\usepackage{bbm} 
\usepackage{euscript}  
\usepackage{authblk}

 \usepackage{mathtools}

\def\@abssec#1{\vspace{.05in}\footnotesize \parindent .2in 
{\bf #1. }\ignorespaces} 
\graphicspath{{/EPSF/}{Figures/}}
\newtheorem{theorem}{Theorem}[section]
\newtheorem{lemma}[theorem]{Lemma}
\newtheorem{proposition}[theorem]{Proposition}
\newtheorem{corollary}[theorem]{Corollary} 

\newtheorem{remark}[theorem]{Remark}

\def \Rm {\mathbb R}
\def \Nm {\mathbb N}

\newcommand{\eps}{\varepsilon}
\newcommand{\vs}{\varrho_\star}

\newcommand{\ds}{\displaystyle}

\newcommand{\calQ}{\mathcal Q}

\newcommand{\calH}{\mathcal H}
\newcommand{\calL}{\mathcal L}
\newcommand{\calK}{\mathcal K}

\newcommand{\calE}{\mathcal E}
\newcommand{\Tr}{\textnormal{Tr}}

\newcommand{\calJ}{\mathcal J}

\newcommand{\calA}{\mathcal A}
\newcommand{\calM}{\mathcal M}

\def\fref#1{{\rm (\ref{#1})}}

\newcommand{\cout}[1]{}

\newcommand{\be}{\begin{equation}}
\newcommand{\ee}{\end{equation}}
\newcommand{\bea}{\begin{eqnarray}}
\newcommand{\eea}{\end{eqnarray}}
\newcommand{\bee}{\begin{eqnarray*}}
\newcommand{\eee}{\end{eqnarray*}}
\newcommand{\bal}{\begin{align*}}
                    \newcommand{\eal}{\end{align*}}
                  \newcommand{\argmin}{\textrm{argmin} \, }

\def\H2p{{H}^2}

\DeclarePairedDelimiter\bra{\langle}{\rvert}
\DeclarePairedDelimiter\ket{\lvert}{\rangle}
\begin{document}
{\title{Constrained minimizers of the von Neumann entropy and their characterization}}

 \author{Romain  Duboscq \footnote{Romain.Duboscq@math.univ-tlse.fr}}
 \affil{Institut de Math\'ematiques de Toulouse ; UMR5219\\Universit\'e de Toulouse ; CNRS\\INSA, F-31077 Toulouse, France}
 \author{Olivier Pinaud \footnote{pinaud@math.colostate.edu}}
 \affil{Department of Mathematics, Colorado State University\\ Fort Collins CO, 80523}

\maketitle

\begin{abstract}
We consider in this work the problem of minimizing the von Neumann entropy
under the constraints that the density of particles, the current, and the kinetic energy of the system is fixed at each point of space. The unique minimizer is a self-adjoint positive trace class operator, and our objective is to characterize its form. We will show that this minimizer is solution to a self-consistent nonlinear eigenvalue problem. One of the main difficulties in the proof is to parametrize the feasible set in order to derive the Euler-Lagrange equation, and we will proceed by constructing an appropriate form of perturbations of the minimizer. The question of deriving quantum statistical equilibria is at the heart of the quantum hydrodynamical models introduced by Degond and Ringhofer in \cite{DR}. An original feature
of the problem is the local nature of constraints, i.e. they depend on position, while
more classical models consider the total number of particles, the total current and the total energy in the system to be
fixed.
\end{abstract}
%\tableofcontents

\section{Introduction}
This work is concerned with the study of minimizers of quantum entropies, which are solutions to problems of the form
\be \label{probmin}
\min_{\calA} \Tr\big( s(\varrho)\big),
\ee
where $\varrho$ is a density operator (i.e. a self-adjoint trace class positive operator), $s$ an entropy function, typically the Boltzmann entropy $s(x)=x \log(x)-x$ for $x\geq 0$, and $\Tr(\cdot)$ denotes operator trace. The feasible set $\calA$ includes linear constraints on $\varrho$ involving particles density, current, and energy. 

This problem is motivated by a series of papers by Degond and Ringhofer on the derivation of quantum hydrodynamical models from first principles, see \cite{DGMRlivre,QHD-CMS,isotherme,QDD-JCP}. It is also a problem arising in the work of Nachtergaele and Yau in their derivation of the Euler equation from quantum dynamics \cite{nachter}. In \cite{DR}, Degond and Ringhofer main idea is to transpose to the quantum setting the entropy closure strategy that Levermore used for kinetic equations \cite{levermore}. The kinetic formulation starts with the transport equation (all physical constants are set to one),
\be \label{boltz}
\partial_t f + \{H,f\}=Q(f), \qquad f\equiv f(t,x,p), \quad (x,p) \in \Rm^d \times \Rm^d,
\ee
where $f \geq 0$ is the particle distribution function, $H$ is a classical Hamiltonian, e.g. $H(x,p)=|p|^2/2+V(x)$ for some potential $V$, $\{H,f\}=\nabla_p H \cdot \nabla_x f-\nabla_x H \cdot \nabla_p f$ is the Poisson bracket, and $Q$ a collision operator. Fluids models are obtained by considering the quantities
$$
\left\{
 \begin{array}{lll}
\ds n_f(t,x)&=&\ds \int_{\Rm^d} f(t,x,p)dp, \qquad  u_f(t,x)=n_f(t,x)^{-1} \int_{\Rm^d} p f(t,x,p)dp\\[3mm]
 \ds k_f(t,x)&=&\ds \frac{1}{2}\int_{\Rm^d} |p|^2 f(t,x,p)dp,
\end{array}
\right.
$$
which are respectively the average particle density, velocity, and kinetic energy. It is not possible to derive a closed system on $n_f$, $u_f$, and $k_f$ from \fref{boltz}, and Levermore's method consists in replacing $f$ in the non closed terms by a statistical equilibrium $f_{\rm{eq}}$. The form of the latter depends on $Q$, and is in some situations the minimizer of the classical Boltzmann entropy
$$
S_c(g)=\int_{\Rm^{2d}} s(g(x,p)) dx dp, \quad \textrm{under the constraints}\quad
n_g=n_f, \quad u_g=u_f, \quad k_g=k_f.
$$
The solution to the minimization problem is the standard Maxwellian
\be \label{maxw}
f_{\rm{eq}}(t,x,p)= \frac{n_f(t,x)}{(2 \pi \, T(t,x))^{d/2}} e^{-\frac{(p-u_f(t,x))^2}{2 T(t,x)}},
\ee
where the temperature $T$ is such that
$$
k_f(t,x)= \frac{1}{2} n_f(t,x)|u_f(t,x)|^2+\frac{d}{2} n_f(t,x) T(t,x).
$$
The equilibrium $f_{\textrm{eq}}$ is an explicit and local function of the constraints, and one may therefore remove the variable $x$ in the definition of the classical entropy $S_c$ and consider the constraints on $n_g$, $u_g$ and $k_g$ to be simple numbers independent of $x$. The well-posedness of the classical minimization problem was addressed in \cite{junk}.

Degond and Ringhofer theory is the quantum version of the above kinetic problem, and their starting point is the quantum Liouville-BGK equation of the form
  \be \label{liouville}
i  \partial_t \varrho=[\calH, \varrho] +\frac{i }{\tau} \left( \varrho_{\rm{eq}}[\varrho]-\varrho \right) ,
\ee
where $\calH$ is a given Hamiltonian, $[\cdot,\cdot]$ denotes the commutator between two operators, $\tau$ is a relaxation time, and $\varrho_{\rm{eq}}[\varrho]$ a quantum statistical equilibrium. The latter is solution to \fref{probmin} under constraints of density, current, and energy as in the classical case. The constraints are defined as follows (we give further an equivalent definition better suited for the mathematical analysis): to any density operator $\varrho$, we can associate a Wigner function $W[\varrho](x,p)$, see e.g. \cite{LP}, so that the particle density $n[\varrho]$, the current density $n[\varrho] u[\varrho]$, and the energy density $w[\varrho]$ of $\varrho$ are given by similar formulas as in the classical picture: 
$$
\left\{
 \begin{array}{l}
\ds n[\varrho](x)=\int_{\Rm^d} W[\varrho](x,p) dp, \quad n[\varrho](x) u[\varrho](x)=\int_{\Rm^d} p W[\varrho](x,p) dp\\[3mm]
\ds w[\varrho](x)=\frac{1}{2}\int_{\Rm^d} |p|^2 W[\varrho](x,p) dp.
\end{array}
\right.
$$
In the statistical physics terminology, fixing the density, current, and energy amounts to consider equilibria in the microcanonical ensemble. The feasible set $\calA$ in \fref{probmin} then consists in density operators $\varrho$ such that $n[\varrho]$, $u[\varrho]$, and $w[\varrho]$ are given functions. Note that the constraints are \textit{local} in $x$, and not \textit{global} as is often found in the literature, see e.g. \cite{Dolbeault-Loss}. In other words, the number of particles (as well as the current and energy) is fixed at each point $x$ of space, rather than prescribing the total number of particles in the system. 

Our main objective in this work is to derive a representation formula such as the Maxwellian \fref{maxw} for the minimizers solution to \fref{probmin}. The problem is considerably more difficult than in the classical case since the solution is now an operator, which depends nonlocally and implicitly on the constraints. The formal solution $\varrho_\star$ to \fref{probmin} reads
\be \label{quantM}
\varrho_\star=e^{-\calH_\star},
\ee
for an appropriate self-consistent Hamiltonian $\calH_\star\equiv \calH[\varrho_\star]$ which depends on the solution $\varrho_\star$. Our main result is a rigorous formulation of \fref{quantM}. We will show that the eigenvalues and eigenfunctions of $\vs$ are the solutions to a nonlinear self-consistent eigenvalue problem. The fact that \fref{probmin} admits a unique solution under local constraints of density, current and energy is established in \cite{DP-JMPA} for a bounded one-dimensional spatial domain. The $\Rm^d$ case is still an open problem, and we will therefore only focus here on the characterization of the minimizer in 1D. The justification of \fref{quantM} in the context of a density constraint only (i.e. only $n[\varrho]$ is prescribed and not $u[\varrho]$ and $w[\varrho]$) was done in \cite{MP-JSP} in a 1D bounded domain, and later extended to $\Rm^d$ in \cite{DP-JFA}. The existence and uniqueness of a minimizer solution to \fref{probmin} under constraints of density and current was established in \cite{MP-KRM} in $\Rm^d$. The fact that the energy constraint is harder to handle than the two others is related to compactness issues, see \cite{DP-JMPA}. Regarding the quantum Liouville-BGK equation, it is shown in \cite{MP-JDE} that \fref{liouville} admits a solution in 1D when the equilibrium is obtained under a density constraint. The diffusion model obtained in the limit $\tau \to 0$ is studied in \cite{pinaudPoincare}. See \cite{jungel-matthes-milisic,jungel-matthes} for more references on quantum hydrodynamics.

One of the main difficulties in the characterization of the minimizer is to properly parametrize the feasible set in order to derive the Euler-Lagrange equation. Indeed, operators in the feasible set must be (i) self-adjoint, (ii) positive, (iii) trace class, and have fixed local (iv) density, (v) current, and (vi) energy. Items (i) and (iii) are somewhat direct to enforce, while (v) is easily handled by a change of gauge (at least in 1D). Items (ii)-(iv)-(vi) together are the most difficult to satisfy. In particular,  additive perturbations found in standard differential calculus only provide here inequalities because of (ii). We will then construct a fine parameterization of the feasible set by choosing perturbations of the minimizer that are both additive and multiplicative, and by using the implicit function theorem in Banach spaces to conclude.

The paper is structured as follows: in Section \ref{secmain}, we introduce the setup and state our main result. Section \ref{secproof} is devoted to the proof of our main theorem and is divided into four steps. In Section \ref{otherproof}, we give some proofs that were postponed in the previous section.
\paragraph{Acknowledgement.} OP is supported by NSF CAREER grant DMS-1452349.

\section{Main result} \label{secmain}

Before stating our main result, we need to introduce some notation and the functional setting.

\paragraph{Notation.} Our spatial domain is $\Omega = [0,1]$. We will denote by $L^r$ , $r\in [1,\infty]$, the usual Lebesgue spaces of complex-valued functions on $\Omega$, and by $W^{k,r}$ the standard Sobolev spaces. We introduce as well $H^k=W^{k,2}$, and $\langle\cdot,\cdot\rangle$ for the Hermitian product on $L^2$ with the convention $\langle f,g\rangle=\int_\Omega f^* g dx$. We will use the notations $\nabla=d/dx$ and $\Delta=d^2/dx^2$ for brevity. The free Hamiltonian $-\frac{1}{2}\Delta$ is denoted by $H_0$, with domain
$$
H_{\textrm{Neu}}^2=\left\{u\in H^2:\,\nabla u(0)=\nabla u(1) = 0\right\}.
$$
Moreover, $\calL(L^2)$ is the space of bounded operators, $\calJ_1 \equiv \calJ_1(L^2)$ is the space of trace class operators,  and $\calJ_2$ the space of Hilbert-Schmidt operators,  all on $L^2$. $\Tr(\cdot)$ denotes operator trace. In the sequel, we will refer to a density operator as a positive, trace class, self-adjoint operator on $L^2$. For $\varrho^*$ the adjoint of $\varrho$ and $|\varrho|=\sqrt{\varrho^* \varrho}$, we introduce the following space:
$$\calE=\left\{\varrho\in \calJ_1:\, \overline{\sqrt{H_0}|\varrho|\sqrt{H_0}}\in \calJ_1\right\},$$
where $\overline{\sqrt{H_0}|\varrho|\sqrt{H_0}}$ denotes the extension of the operator $\sqrt{H_0}|\varrho|\sqrt{H_0}$ to $L^2$. The domain of $\sqrt{H_0}$ is $H^1$. We will often drop the extension sign in the sequel to ease notation, and  keep it when it is relevant.  The space $\calE$ is a Banach space when endowed with the norm
$$\|\varrho\|_{\calE}=\Tr \big(|\varrho| \big)+\Tr\big(\sqrt{H_0}|\varrho|\sqrt{H_0}\big).$$
The energy space is the following closed convex subspace of $\calE$:
$$\calE^+=\left\{\varrho\in \calE:\, \varrho\geq 0\right\}.$$
The eigenvalues of a density operator are counted with multiplicity, and form a nonincreasing sequence, converging to zero if the sequence is infinite. The notation $a \lesssim b$ stands for $a \leq C b$, where $C$ is a constant independent of $a$ and $b$.

\paragraph{Setting of the problem.} The first three moments of a density operator $\varrho$ are defined in terms of $\varrho$ and not its Wigner function as follows: for any smooth function $\varphi$ on $\Omega$, and identifying a function with its associated multiplication operator, the (local) density $n[\varrho]$, current $n[\varrho] u[\varrho]$ and energy $w[\varrho]$ of $\varrho$ are uniquely defined by duality by 
\begin{align*}
&\int_\Omega n[\varrho] \varphi dx = \Tr \big( \varrho \varphi \big), \qquad \int_\Omega n[\varrho]u[\varrho] \varphi dx= -i \Tr \left(\varrho \left(\varphi \nabla+\frac{1}{2} \nabla \varphi \right) \right) \\
&\int_\Omega w[\varrho] \varphi dx=  -\frac{1}{2}\Tr \left(\varrho \left( \nabla  \varphi \nabla +\frac{1}{4} \Delta \varphi \right) \right) . 
\end{align*}
Denote by $\{\rho_p,\phi_p\}_{p \in \Nm}$ the spectral elements of a density operator $\varrho$ (the number of nonzero eigenvalues might be finite or not), and let
$$
k[\varrho]=-n[\nabla \varrho \nabla]=\sum_{p \in \Nm} \rho_p |\nabla \phi_p|^2.
$$
A short calculation shows that
$$
w[\varrho]=\frac12\left( k[\varrho]-\frac{1}{4} \Delta n[\varrho]\right).
$$
Hence, since $n[\varrho]$ is prescribed, we can equivalently set a constraint on $w[\varrho]$ or on $k[\varrho]$, and we choose $k[\varrho]$ since it is positive. Note also the classical (formal) relations
\be \label{defcurrent}
n[\varrho ]= \sum_{p \in \Nm} \rho_p |\phi_p|^2, \qquad  j[\varrho]=n[\varrho ] u[\varrho ]=\Im \left( \sum_{p \in \Nm} \rho_p \phi_p^* \nabla \phi_p \right).
\ee
\begin{remark} \label{rem1} Let $\varrho \in \calE^+$ with eigenvalues $\{\rho_p\}_{p \in \Nm}$ and eigenvectors $\{\phi_p\}_{p \in \Nm}$. Then
  $$
  n[\varrho]=\sum_{p \in \Nm} \rho_p |\phi_p|^2, \qquad k[\varrho]=\sum_{p \in \Nm} \rho_p |\nabla \phi_p|^2,
  $$
  with convergence in $L^1$ and almost everywhere. The density $n[\varrho]$ is bounded since $\varrho \in \calE^+$ implies that $\nabla \sqrt{n[\varrho]} \in L^2$ according to the inequality
$$
\|\nabla \sqrt{n[\varrho]}\|^2_{L^2} \leq \|\varrho\|_{\calE},
$$ and a Sobolev embedding gives $n[\varrho] \in L^\infty$. 
\end{remark}

For $\varrho \in \calE^+$, we denote the entropy of $\varrho$ by
\begin{equation*}
S(\varrho)=\Tr(s(\varrho)),
\end{equation*}
where $s(x) = x \log  x -  x$ is the Boltzmann entropy. The entropy $S$ is referred to as the von Neumann entropy. We define the set of admissible contraints
 $$
\calM= \left\{ (n_0,u_0,k_0) \in \left(L^1\right)^3: \; n_0=n[\varrho],\, u_0=u[\varrho],\, k_0=k[\varrho],\, \textrm{for some } \varrho \in \calE^+ \right\},
$$
that is the set of functions $(n_0,u_0,k_0)$ such that there is at least one density operator in $\calE$ with local density, current, and kinetic energy given by $(n_0,u_0,k_0)$. The structure of $\calM$ is unknown as of now, but this is not an issue for us since in our problem of interest, the constraints are always in $\calM$ as originating from the solution $\varrho$ to the quantum Liouville-BGK equation \fref{liouville}. For the constraints $n_0,u_0$ and $k_0$ in $\calM$, we then define the feasible set 
\begin{equation*}
\calA(n_0,u_0,k_0) = \left\{ \varrho\in\calE^+:\;  n[\varrho] = n_0, \; u[\varrho] = u_0, \; k[\varrho] = k_0\right\}.
\end{equation*}
\begin{remark} \label{rem2} Constraints in the admissible set $\calM$ satisfy some compatibility conditions. 
Let indeed $\varrho\in\calE^+$. Then, 
\begin{equation} \label{ineqk}
k[\varrho]\geq n[\varrho] |u[\varrho]|^2+\left|\nabla \sqrt{n[\varrho]}\right|^2.
\end{equation}
To see this, we remark that 
\begin{align} \label{nabn}
\frac{1}{2}\nabla n[\varrho] = \sum_{p\in\mathbb{N}} \rho_p \Re\left( \phi_p^* \nabla \phi_p\right)\quad\textrm{and}\quad j[\varrho] = n[\varrho]u[\varrho] = \sum_{p\in\mathbb{N}} \rho_p \Im\left( \phi_p^* \nabla \phi_p\right),
\end{align}
where both series converge in $L^1$ and a.e. according to Remark \ref{rem1},
and, thus, we deduce that
\begin{equation*}
\frac{1}{2}\nabla n[\varrho] + i j[\varrho] = \sum_{p \in \Nm} \rho_p \phi_p ^* \nabla \phi_p.
\end{equation*}
It follows that, by the Cauchy-Schwarz inequality,
\begin{align*}
\frac{1}{4}|\nabla n[\varrho]|^2+ |j[\varrho]|^2  &= \left| \frac{1}{2}\nabla n[\varrho] + i j[\varrho]\right|^2 = \left| \sum_{p \in \Nm} \rho_p \phi_p ^* \nabla \phi_p\right|^2
\\ &\leq \left(\sum_{p \in \Nm}\rho_p |\phi_p|^2\right)\left(\sum_{p \in \Nm}\rho_p |\nabla \phi_p|^2\right) = n[\varrho]k[\varrho],
\end{align*}
which yields \fref{ineqk}.
\end{remark}

The fact that $S$ admits a unique minimizer in $\calA(n_0,u_0,k_0)$ was proven in \cite{DP-JMPA}. The result is the following:
\begin{theorem} \label{thexist} Suppose that $(n_0,u_0,k_0) \in \calM$, where $\Delta n_0 \in L^2$ with $n_0>0$, $u_0 \in L^2$. Then, the constrained minimization problem
  $$
  \min_{\calA(n_0,u_0,k_0)} S(\varrho)
  $$
  admits a unique solution. If $\|k_0\|_{L^1}=\|\nabla \sqrt{n_0}\|^2_{L^2}+\|\sqrt{n_0} u_0\|^2_{L^2}$, then the solution is
  $$
  e^{i \int_0^x u_0(y) dy} | \sqrt{n_0} \rangle \langle \sqrt{n_0} | e^{-i \int_0^x u_0(y) dy}.$$
\end{theorem}

Note that Theorem \ref{thexist} was obtained in \cite{DP-JMPA} for periodic boundary conditions. The proof immediately generalizes to the Neumann conditions that we chose here since they somewhat simplify some technicalities. The main result of this paper is the characterization of the minimizer of Theorem \ref{thexist}. Since the spatial domain is one-dimensional, we can actually treat the current constraint by a simple change a gauge and not consider it in the minimization. Suppose indeed that we can characterize the minimizer of $S$, denoted $\vs$, in the set
\begin{equation*}
\calA(n,k) = \left\{ \varrho\in\calE^+:\;  n[\varrho] = n, \; k[\varrho] = k\right\},
\end{equation*}
where $(n,k) \in \calM_0$ with
$$
\calM_0= \left\{ (n_0,k_0) \in \left(L^1 \right)^3: \; n_0=n[\varrho],\, k_0=k[\varrho],\, \textrm{for some } \varrho \in \calE^+ \right\}.
$$
Suppose in addition that $u[\vs]=0$. Then, we claim that the minimizer in $\calA(n_0,u_0,k_0)$, for $n = n_0$ and $k = k_0-n_0 u_0^2$, is
$$\widetilde{\vs}= e^{i \int_0^x u_0(y)dy} \varrho_\star \, e^{-i \int_0^x u_0(y)dy}.$$
We have indeed, for any $\varrho \in \calE^+$,
\bea %\label{rel1}
j\left[e^{i \int_0^x u_0(y)dy}\, \varrho \, e^{-i \int_0^x u_0(y) dy}\right]&=&j[\varrho]+n[\varrho]u_0 \nonumber\\
k\left[e^{i \int_0^x u_0(y)dy} \, \varrho \, e^{-i \int_0^x u_0(y) dy}\right]&=&k[\varrho]+n[\varrho]u_0^2+2 n[\varrho] u[\varrho]u_0, \label{rel2}
\eea
and therefore, since $u[\varrho_\star]=0$, it follows that $j[\widetilde{\vs}]=n_0 u_0$ and $k[\widetilde{\vs}]=k_0$. This shows that $\widetilde{\vs}$ satisfies the constraints $(n_0,u_0,k_0)$. It is the minimizer since on the one hand,
$$
  \min_{\calA(n,0,k)} S(\varrho)=\min_{\calA(n,k)} S(\varrho)=S(\vs),
  $$
  since $\calA(n,0,k) \subset \calA(n,k)$ and $u[\vs]=0$, and on the other, denoting for the moment by $\sigma_\star$ the minimizer in $\calA(n_0,u_0,k_0)$,
  $$S(\widetilde{\vs}) \geq S(\sigma_\star)=S(e^{-i \int_0^x u_0(y)dy}\sigma_\star e^{i \int_0^x u_0(y)dy}) \geq \min_{\calA(n,0,k)} S(\varrho)=S(\vs)=S(\widetilde{\vs}).
  $$
  Above, we used that unitary equivalent operators have the same eigenvalues and therefore the same entropy, and that $e^{-i \int_0^x u_0(y)dy}\sigma_\star e^{i \int_0^x u_0(y)dy} \in \calA(n,0,k)$. Hence,
 $$ S(\widetilde{\vs}) = S(\sigma_\star),$$
 and since the minimizer is unique, we conclude that $\widetilde{\vs}=\sigma_\star$. Note that $(n_0,u_0,k_0) \in \calM$ implies that $(n,k) \in \calM_0$. Indeed, if $(n_0,u_0,k_0)$ are the moments of $\varrho_0 \in \calE^+$, then $(n,k)$ are the density and energy of $e^{-i \int_0^x u_0(y)dy} \varrho_0 e^{i \int_0^x u_0(y)dy}$ according to \fref{rel2}.

 From now on, we only consider the minimization problem in $\calA(n,k)$. We make the following assumptions on $(n,k) \in \calM_0$.

\paragraph{Assumptions A.}%\label{hyp:bnd}
Let
\begin{equation}\label{eq:defa}
a(x) = \left(k(x) - \left|\nabla\sqrt{n(x)} \right|^2\right)^{-1},
\end{equation}
and denote
\begin{equation*}
n_m = \min_{x\in [0,1]} n(x).
\end{equation*}
We assume that
\begin{enumerate}
\item $\sqrt{n}\in H^{1}$, $\Delta n \in L^2$, and $n_m > 0$,
\item $k \in L^{\infty}$,
\item there exists $a_M>0$ such that $a^{-1}(x) > a_M^{-1} >0$ a.e.
\end{enumerate}

Under Assumptions A, $S$ admits a unique minimizer $\vs$ in $\calA(n,k)$ (this is a direct adaptation of Theorem \ref{thexist}). And because of item 3, this minimizer is not simply $| \sqrt{n_0} \rangle \langle \sqrt{n_0} |$.

\paragraph{Main result.} We introduce first the following, for $\{\rho_p\}_{p \in \Nm}$ and $\{\phi_p\}_{p \in \Nm}$ the eigenvalues and eigenvectors of $\vs$:
\begin{align} \label{defK}
K_0(x,y)&=2 \Re \sum_{p \in \Nm} \rho_p \phi_p^*(x) \nabla \phi_p(y), \hspace{0.3cm} K(x,y)=\frac{K_0(x,y)}{2 n(x)} +\frac{a(x) \nabla n(x)}{4 n(x)}\nabla_x K_0(x,y),
\end{align}
for any $(x,y)\in\Omega \times \Omega$, where $a$ is defined in \fref{eq:defa} and $n$ is the constraint. Note that the series defining $K_0$ and $\nabla_x K_0$ converge almost everywhere according to Remark \ref{rem1}, and that we have the estimates $|K_0(x,y)| \leq 2 \sqrt{n(x) \, k(y)}$ and $|\nabla _x K_0(x,y)| \leq 2 \sqrt{k(x) \, k(y)}$, $x,y$ a.e. in $\Omega \times \Omega$. Since both $n$ and $k$ are bounded according to Assumptions A, it follows that $K_0$ and $\nabla_x K_0$ belong to $L^\infty \times L^\infty$. Moreover, since $a$, $\nabla n$ (since $k \in L^\infty$) and $n^{-1}$ are bounded according to Assumptions A, it follows that $K \in L^\infty \times L^\infty$. For an arbitrary kernel $N \in L^2 \times L^2$, we then define the integral operator $\calL_N$ and its adjoint $\calL_N^*$ by, for all $\varphi \in L^2$ and for any $x\in\Omega$,
\begin{equation}\label{defL}
\calL_{N}\varphi(x) = \int_0^x N(x,y) \varphi(y)dy, \qquad \calL^*_{N}\varphi(x) =\int_x^1 N(y,x) \varphi(y)dy.
\end{equation}
Let also, for every $x\in\Omega$,
\begin{align} \label{defgamma}
 \gamma_{\star}(x)&= 2 \Re \sum_{ p \in \Nm} \rho_{p} \nabla \phi_{p}(x)  \int_x^1 \phi_{p}^* (y) \left(\log(\rho_{p} )-\frac{n [\varrho_{\star} \log(\varrho_{\star})](y)}{n(y)}\right) dy,
\end{align}
which will be proved to belong to $L^\infty$, and let $m_\star =  a\, m_0 /2\in L^\infty$, where $m_0$ is the unique solution to the adjoint equation
\be \label{adj}
m_0=\calL_K^*m_0+ \gamma_\star.
\ee
The facts that the equation above admits a unique solution in $L^\infty$and that $m_\star$ is positive will be estalished in Sections \ref{proofpropdire} and \ref{sec:conc}. For $\varphi, \psi \in H^1$, consider finally the sesquilinear form
%\begin{equation}\label{eq:SesquilinearF}
$$
\mathcal{Q}_{\star}(\psi,\varphi) = \int_0^1 n(x) \left(\nabla \left( \frac{\psi^*(x)}{\sqrt{n(x)}}\right) \nabla \left( \frac{\varphi(x)}{\sqrt{n(x)}}\right)\right)m_\star(x)  dx + \int_0^1 A_{\star}(x)  \psi^*(x) \varphi(x) dx,
$$
%\end{equation}
where 
\begin{equation*}
A_{\star} = -\frac{n[\varrho_\star \log(\varrho_\star)]}{n}-\frac{m_\star}{n}k.
\end{equation*}
We will prove further that $A_\star \in L^\infty$. That $\calQ_\star$ is well-defined on $H^1$ is a consequence of the facts that $n$ is bounded below and that $\nabla n \in L^\infty$.
%
%We will use the following normed space: 
%\begin{equation*}
%H_\star^1=\left\{\varphi \in L^2: \|\varphi\|^2_{H^1_\star} = \int_0^1 |\nabla \varphi(x)|^2 \,m_\star(x) \,dx+\int_0^1 |\varphi(x)|^2dx< \infty \right\}.
%\end{equation*}
Our main result is the following:

\begin{theorem} \label{mainth} Let $(n,k) \in \calM_0$ satisfy Assumptions A, and let $\varrho_\star$ be the unique minimizer of $S$ in $\calA(n,k)$. Denote by $\{\rho_p\}_{p\in\mathbb{N}}$ and $\{\phi_p\}_{p\in\mathbb{N}}$ the eigenvalues and eigenfunctions of $\varrho_\star$. Then $\varrho_\star$ is full rank, i.e. $\rho_p>0$ for all $p \in \Nm$, and $\{\rho_p\}_{p\in\mathbb{N}}$ and $\{\phi_p\}_{p\in\mathbb{N}}$ verify the self-consistent nonlinear eigenvalue problem
\be \label{minieig}
- \log (\rho_p) = \min_{\varphi \in \calK_p} \calQ_\star(\varphi,\varphi)=\calQ_\star(\phi_p,\phi_p), \qquad p \in \Nm,
\ee
where
$$
\calK_p=\left\{ \varphi \in H^1: \|\varphi\|_{L^2}=1\quad\textrm{and}\quad \varphi\in \left(\textrm{span}\{\phi_j\}_{0\leq j\leq p-1} \right)^{\perp}\right\},
$$
with the convention  $\calK_0=\{ \varphi \in H^1, \;  \|\varphi\|_{L^2}=1\}$. Morever, the current $n u[\vs]$ carried by $\vs$ vanishes. 
\end{theorem}

Theorem \ref{mainth} provides us with a rigorous formulation of the relation \fref{quantM}, where $\calH_\star$ is formally
$$
\calH_\star=-\frac{1}{\sqrt{n}} \nabla   \left( n\, m_\star\nabla \left( \frac{\cdot}{\sqrt{n}}\right) \right)  + A_\star.
$$
Note that the structure of the form $\calQ_\star$ alone does not allow us to conclude that $\calQ_\star$ admits minimizers on the subspaces $\calK_p$. The reason for this is that we only have very minimal information on $m_\star$: we only know that $m_\star \geq 0$ and that $m_\star \in L^\infty$, and in particular there is no sufficient information to both use the min-max principle and to prove that the weighted space
\begin{equation*}
H_\star^1=\left\{\varphi \in L^2: \int_0^1 |\nabla \varphi(x)|^2 \,m_\star(x) \,dx+\int_0^1 |\varphi(x)|^2dx< \infty \right\}
\end{equation*}
is complete. It is unclear at this point if improved regularity of the data $(n,k)$ translates to more regularity on $m_\star$. In the same way, the fact that $m_\star=am_0/2$ is positive is not a consequence of $m_0$ solving \fref{adj}, it follows from the minimization problem, more precisely from the fact that $\calQ_\star$ is bounded from below as a consequence of the Euler-Lagrange equation. To summarize, the properties of $\calQ_\star$ and $m_\star$ are all inherited from the original minimization problem, and are difficult, if possible, to establish alone.

Theorem \ref{mainth} can be generalized to other entropies, in particular to the Fermi-Dirac entropy of the form $s(x)= x \log(x)+(1-x) \log(1-x)$, which shares the same technical difficulties as the Boltzmann entropy.

The rest of the paper consists of the proof of Theorem \ref{mainth}.
% Let the space $\frakH$ be defined by
% $$ \frakH=\left\{ \varphi \in L^2: - \sum_{j \in \mathbb{N}} \log (\rho_j)|(\phi_j,\varphi)|^2 < \infty \right\},$$
% where we recall  that $\{\rho_j\}_{j \in \mathbb{N}}\subset(0,1]$ is the nonincreasing sequence of eigenvalues of $\varrho_{\star}$ and $\{\phi_j\}_{j \in \mathbb{N}}$ its eigenvectors, which form an orthonormal basis of $L^2$. The space $\frakH$ is a Hilbert space when equipped with the inner product
% $$
% (u,v)_\frakH=-\sum_{j \in \mathbb{N}} \log (\rho_j)(\phi_j,u)^* (\phi_j,v).
% $$
% Note that $\frakH$ is the domain of self-adjointness of $\sqrt{- \log(\varrho_{\star})}$. As a corollary of Theorem \ref{mainth}, we have the following result.
% \begin{corollary}\label{cor:main}
% By denoting $\mathcal{Q}_{m,S}$ the restriction of $\mathcal{Q}_m$ to $S=$span$\{\phi_j\}_{j\in \mathbb{N}}$, we have that $\mathcal{Q}_{m,S}$ is densely defined and closable, and that $- \log (\varrho_\star)$ is the unique self-adjoint operator associated with the closure $\overline{\calQ}_{m,S}$. Finally, $H^1_m \subset \mathfrak{H}$.
% \end{corollary}

\section{Proof of Theorem \ref{mainth}} \label{secproof}

\paragraph{Outline of the proof.} The proof is divided into four steps. In the first step, we construct a parameterization of the feasible set $\calA(n,k)$ in order to perturb around the minimizer and obtain the Euler-Lagrange equation. This is done by using the implicit function theorem and an appropriate class of perturbations. The second step is the derivation of the Euler-Lagrange equation. There is a technical difficulty since the derivative of the entropy $s'(x)=\log(x)$ is singular at zero. We will therefore regularize the entropy and then pass to the limit. The third step consists in proving that $\vs$ is full rank, and is based on a proper use of the Euler-Lagrange equation. In the fourth step, we finally establish the positivity of $m_\star$ and the minimization principle \fref{minieig} of Theorem \ref{mainth}.

We start by introducing some notation. For $f\in L^2$, consider the bounded operator $T_f\,:\, H^1 \to L^{\infty}$, defined by
$$
T_f u (x):=\int_0^x f(y) \nabla u(y) dy.
$$
Furthermore, for $\varphi \in W^{1,\infty}$, let $\varrho_1=\sqrt{\rho_p}\ket{\varphi} \bra{\phi_p}$ (we use the standard Dirac bra-ket notation), where $\phi_p$ is the eigenvector of the minimizer $\vs$ associated with the eigenvalue $\rho_p$. With $\varrho_2\in\mathcal{E}^+$,  $f=(f_2,f_3)\in L^\infty \times L^\infty$, $t \in [-1,1]$, we define the operator
\begin{equation*}
L(t,f)=t \varrho_1+ f_2 I + T_{f_3}, \qquad \textrm{with domain} \qquad D(L(t,f))=H^1,
\end{equation*}
and introduce 
%\begin{equation}\label{eq:varrhotf}
$$
\varrho(t,f)=(I+L(t,f))(\vs + t \varrho_2)(I+L^*(t,f)),
$$
%\end{equation}
where $I$ the identity operator and $L^*(t,f)$ the adjoint of $L(t,f)$ (formal at that stage).

The rationale behind the choice of $L(t,f)$ and $\varrho(t,f)$ is the following: we need first $\varrho(t,f)$ to be a density operator, and  if $t \varrho_2$ is positive, it is clear that $\varrho(t,f)$ is positive (and therefore self-adjoint). The trace class property is a consequence of the regularity of $f$ and $\varrho_1$ and will be established further.

We need moreover to impose the density and energy constraints on the two functions $n[\varrho(t,f)]$ and $k[\varrho(t,f)]$, and with the goal of using the implicit function theorem, it is natural to introduce two functions $f_2$ and $f_3$ for this. More precisely, the operator $f_2 I$ for $f_2$ real-valued acts as multiplication of the local density by $f_2$ since $n[f_2 \varrho]=f_2 n[\varrho]$ for any density operator $\varrho$; in the same way, $T_{f_3}$ multiplies the local energy by $f_3$ since $k[T_{f_3} \varrho]=f_3 k[\varrho]$. The variable $t$ will allow us to parametrize the feasible set with some $(f_2(t),f_3(t))$ obtained with the implicit function theorem. The operators $\varrho_1$ and $\varrho_2$ serve as ``test operators'' in the Euler-Lagrange equation, and both provide us with independent information. On the one hand, the operator $\varrho_1$ can have an arbitrary sign and leads to test operators of form $\varrho_1 \vs + \vs \varrho_1^*$ and to Lemma \ref{lem:SesquOrth}. The presence of $\vs$ in the previous expression limits what can infered about the form $\calQ_\star$. On the other hand, as an additive positive perturbation, $\varrho_2$ leads to a inequality, with now a test operator independent of $\vs$. This results in particular in Corollary \ref{prop:QineqNN} and in the fact that $\vs$ is full rank.

\subsection{Step 1: Construction of admissible directions}%\label{sec:implicitfun}

The next lemma provides us with a proper definition of $\varrho(t,f)$.
\begin{lemma} \label{proper}Let $\sigma \in \calE^+$ and $f \in L^\infty \times L^\infty$. Then $B=(I+L(t,f)) \sqrt{\sigma}$ is bounded, and $\overline{(I+L(t,f))\sigma (I+L^*(t,f))}=B B^*$.
  \end{lemma}
  \begin{proof} We note first that $\sqrt{\sigma} L^2 \subset H^1$ since $ \| \sqrt{\sigma} \varphi \|_{H^1} \leq  \|\sigma\|_{\calE}^{1/2} \|\varphi\|_{L^2}$ for all $\varphi \in L^2$. Hence, $(I+L(t,f)) \sqrt{\sigma}$ is bounded since $D(L(t,f))=H^1$. Furthermore, a direct calculation shows that, for $\rho_j[\sigma]$ and $\phi_j[\sigma]$ the eigenvalues and eigenfunctions of $\sigma$ and $\varphi \in L^2$,
    $$
    \overline{\sqrt{\sigma} (I+L^*(t,f))} \varphi=\sum_{j \in \Nm} \sqrt{\rho_j[\sigma]} \langle (I+L(t,f))\phi_j[\sigma], \varphi \rangle \phi_j[\sigma],
    $$
    with convergence in $L^2$. It suffices then to identify the expression of the r.h.s above with that of $((I+L(t,f)) \sqrt{\sigma})^* \varphi$ to obtain  $((I+L(t,f)) \sqrt{\sigma})^*=\overline{\sqrt{\sigma} (I+L^*(t,f))}$ . This concludes the proof since $$\overline{(I+L(t,f))\sigma (I+L^*(t,f))}=(I+L(t,f))\sqrt{\sigma}\, \overline{\sqrt{\sigma} (I+L^*(t,f))},$$
    as $(I+L(t,f))\sqrt{\sigma}$ is bounded. 
    % We have
%     $$
% (I+L(t,f)) \sqrt{\sigma} \phi_j=\sqrt{\rho_j} (\phi_j+t \rho_1 \phi_j+ f_2 \phi_j+ T_{f_3} \phi_j)
% $$
%     Let $\varrphi \in L^2$, and write
%     $$
%     \sigma \varphi=\sum_{j \in \Nm} \rho_j \langle \phi_j, \varphi \rangle,
%     $$
%     with convergence in $L^2$. Consider the series
%     $$
%     \xi_N=\sum_{|j| \leq N}  \rho_j \langle \phi_j, \varphi \rangle ((I+L(t,f))) \phi_j.
%     $$
%     Then
%     $$
%     \| \xi_N\|_{L^2}^2 \leq \sum_{|j| \leq N}  \rho_j 
%     $$
    \end{proof}

\bigskip

    Since both $\vs$ and $\varrho_2$ are in $\calE^+$, we can then interpret $\varrho(t,f)$ as
    \be \label{decomp}
    \varrho(t,f)= A_0 A_0^*+t A_1 A_1^*, 
    \ee
    where $A_0=(I+L(t,f)) \sqrt{\vs}$ and $A_1=(I+L(t,f)) \sqrt{\varrho_2}$, both being bounded. 

We remark in passing that
$$
n[\varrho_1\vs + \vs\varrho_1^*]= 2 \rho_p^{3/2}\,\Re \varphi^* \phi_p, \qquad k[\varrho_1\vs + \vs\varrho_1^*]=2 \rho_p^{3/2}\,\Re \nabla \varphi^* \nabla \phi_p, 
$$
which are both bounded since $\varphi \in W^{1,\infty}$ and $\sqrt{\rho_p}\phi_p \in W^{1,\infty}$ since $k \in L^\infty$.

    The main result of this section is the following:
    %We will only consider cases when either $t \in [-1,1]$ and $\varrho_2=0$, or $t \in [0,1]$ and $\varrho_2 \neq 0$. In both cases, the operator $\vs+t \varrho_2$ is positive and the operator $\varrho(t,f)$ is then interpreted in the sense of Lemma \ref{proper}. In particular, $\varrho(t,f)$ is positive. 

\begin{proposition} \label{direct} (Parameterization of the feasible set). Suppose $\varrho_2=0$, (resp. $\varrho_1=0$, $k[\varrho_2] \in L^\infty$). Then, there exists $t_0>0$ and $f \in C^1((-t_0,t_0), W^{1,\infty} \times L^\infty)$ (resp. $C^1([0,t_0), W^{1,\infty} \times L^\infty)$) such that $\varrho(t,f(t)) \in \calA(n,k)$ for all $t \in (-t_0,t_0)$ (resp. $t \in [0,t_0)$). Moreover, the derivatives at $t=0$, $f'_2(0)$ and $f'_3(0)$ verify
\begin{equation}\label{sysf}\left\{\begin{array}{ll}
n_1+ n_2+2 f'_2(0) n+ \calL_{K_0} f_3'(0) =0,
\\ f_3'(0)=\calL_{K} f_3'(0) +b,
\end{array}\right. \end{equation}
%$$
%f_2'(0)=-\frac{n_1}{2n}-\calL_{K_0} f_3'(0), \qquad f_3'(0)=\calL_{K_2} f_3'(0)+b,
%$$
where $\calL_K$, $\calL_{K_0}$ are defined in \fref{defK}-\fref{defL}, $n_1=n[\varrho_1\vs + \vs\varrho_1^*]\in L^\infty$, $k_1=k[\varrho_1\vs + \vs\varrho_1^*] \in L^\infty$, $n_2=n[\varrho_2]\in L^\infty$, $k_2=k[\varrho_2] \in L^\infty$, and 
\begin{equation*}
b = \frac{n_1+n_2}{2n} - \frac{a}{2}\left(k_1+k_2 - \frac{\nabla\sqrt{n}}{\sqrt{n}} \nabla (n_1+n_2)\right).
\end{equation*}
Finally, $\varrho(t,f(t)) \in C^1((-t_0,t_0),\calJ_1)$ (resp. $C^1([0,t_0),\calJ_1))$.
\end{proposition}

The proof of Proposition \ref{direct} is fairly long and is postponed to Section \ref{proofpropdire}.

\subsection{Step 2: The Euler-Lagrange equation}%\label{sec:EulerLagrange}

\paragraph{Regularization.} We have constructed a proper parameterization of the feasible set in the last section, and are now in position to derive the Euler-Lagrange equation. There is a technical issue since the derivative of the entropy is singular at $x=0$, and it is unclear how to proceed without regularization. We then consider the smoothed entropy, for $\eta\in(0,1]$, 
\begin{equation*}
s_\eta(x) = (x+\eta)\log(x + \eta) - x - \eta \log(\eta),\quad x\in\mathbb{R}^+,
\end{equation*}
with $S_{\eta}(\varrho) = \Tr(s_{\eta}(\varrho))$ for $\varrho \in \calE^+$. By adapting the techniques of \cite{DP-JMPA} used in the proof of Theorem \ref{thexist}, it can be shown that $S_\eta$ admits a unique minimizer $\varrho_{\star,\eta}$ in $\mathcal{A}(n,k)$, and we denote by $\{\rho_{p,\eta}\}_{p\in\mathbb{N}}$ the nonincreasing sequence of eigenvalues of $\varrho_{\star,\eta}$, and by $\{\phi_{p,\eta}\}_{p\in\mathbb{N}}$ the associated eigenfunctions.

Proposition $\ref{direct}$ applies to $\varrho_{\star,\eta}$, and we denote by $\varrho_\eta(t,f(t))$ the corresponding perturbed operator with $\varrho_{1}=\varrho_{1,\eta}=\sqrt{\rho_{p,\eta}}\ket{\varphi} \bra{\phi_{p,\eta}}$ (note that $f$ depends on $\eta$ but this fact is omitted to alleviate notation). We set $\varrho_2=0$ in $\varrho_\eta(t,f(t))$ until further notice, and then consider the function $F_{\eta} : (-t_0,t_0)\to \mathbb{R}$ given by ($t_0$ is the one coming from the version of Proposition \ref{direct} for $\varrho_{\star,\eta}$ and depends on $\eta$),
\begin{equation*}
F_{\eta}(t)=\Tr \big( s_{\eta}(\varrho_{\eta}(t,f(t)))\big).
\end{equation*}

Before deriving the Euler equation for the regularized problem, we state the two lemmas below, proved in Sections \ref{proofE1} and \ref{proofEL1}.
\begin{lemma}\label{estimnlog} The operator $\varrho_{\star,\eta} \log(\varrho_{\star,\eta}+\eta)$ is trace class, and there exists $C$ independent of $\eta$ such that
  $$\|n [\varrho_{\star,\eta} \log(\varrho_{\star,\eta}+\eta)]\|_{L^\infty} \leq C.$$
  \end{lemma}
 
  The second lemma concerns $F_\eta$.%, and below $\rho_{j,\eta}$ and $\phi_{j,\eta}$ denote the eigenvalues and eigenfunctions of $\varrho_{\star,\eta}$.
\begin{lemma}\label{lem:EulerLag1}
The function $F_\eta$ belongs to $C^1((-t_0,t_0))$, and we have
\begin{equation*}
  \left. \frac{d F_{\eta}(t)}{d t} \right|_{t=0}= \Tr \big(\varrho_{\star,\eta} \log(\varrho_{\star,\eta} + \eta) (\varrho_{1,\eta}+\varrho_{1,\eta}^*)\big)- \left\langle \frac{n[\varrho_{\star,\eta} \log(\varrho_{\star,\eta}+\eta)]}{n},n_{1,\eta} \right\rangle+\langle \gamma_{\eta},f'_3(0)\rangle,
\end{equation*}
where $n_{1,\eta}=n[\varrho_1 \varrho_{\star,\eta} + \varrho_{\star,\eta}\varrho_1^*]$,
\begin{align*}
 \gamma_{\eta}(x)&= 2 \Re \sum_{ j \in \Nm} \rho_{j,\eta} \nabla \phi_{j,\eta}(x)  \int_x^1 \phi_{j,\eta}^* (y) \left(\log(\rho_{j,\eta} + \eta)-\frac{n [\varrho_{\star,\eta} \log(\varrho_{\star,\eta}+\eta)](y)}{n(y)}\right) dy,
\end{align*}
and $\gamma_\eta \in L^\infty$.
\end{lemma}

\paragraph{The Euler-Lagrange equation for the regularized problem.}
In the lemma below, $\calL^*_{K_\eta}$ is defined in \fref{defK}-\fref{defL}, with $K_0$ replaced by $K_{0,\eta}$ in \fref{defL}, and $K_{0,\eta}$ is defined as $K_0$ with $\rho_{j,\eta}$ and  $\phi_{j,\eta}$ in place of the eigenvalues and eigenvectors of $\vs$. With Remark \ref{rem1}, we can show that as $K$, $K_\eta$ is bounded, with the estimate
\be \label{estKeta}
 \| K_\eta \|_{L^\infty \times L^\infty} \lesssim \|n\|_{L^\infty}+\|k\|_{L^\infty}.
\ee

\begin{lemma}%\label{lem:EulerLagrangeAb} 
(Euler-Lagrange for the regularized problem).
For the $\gamma_\eta \in L^\infty$ defined in Lemma \ref{lem:EulerLag1}, the adjoint problem  
\begin{equation}\label{eq:SmoothmEq}
m_{0,\eta}=\calL^*_{K_{\eta}} m_{0,\eta} -\gamma_{\eta},
\end{equation}
admits a unique solution $m_{0,\eta} \in L^\infty$. Introducing $m_{\eta}=a m_{0,\eta}/2$, where $a$ is defined in \fref{eq:defa}, we have the Euler-Lagrange equation
\begin{align} \nonumber
  \Tr \big(\varrho_{\star,\eta} \log(\varrho_{\star,\eta}+\eta)&(\varrho_{1,\eta}+\varrho_{1,\eta}^*)\big)+\langle m_{\eta},k_{1,\eta}\rangle+\langle A_{\eta},n_{1,\eta}\rangle \\
  &+ \left\langle m_{\eta}\frac{|\nabla\sqrt{n}|^2}{n},n_{1,\eta}\right\rangle-\left\langle m_{\eta},\frac{\nabla \sqrt{n}}{\sqrt{n}} \nabla n_1\right\rangle=0,\label{eq:EulerLagrangeAb}
\end{align}
where
\begin{equation*}
A_{\eta} = -\frac{n[\varrho_{\star,\eta} \log(\varrho_{\star,\eta} + \eta)]}{n}-\frac{k}{n}m_{\eta},
\end{equation*}
and $n_{1,\eta}=n[\varrho_1 \varrho_{\star,\eta} + \varrho_{\star,\eta}\varrho_1^*]\in L^\infty$, $k_{1,\eta}=k[\varrho_1 \varrho_{\star,\eta} + \varrho_{\star,\eta}\varrho_1^*]\in L^\infty$.
\end{lemma}
\begin{proof}
  Since $\varrho_{\eta}(t,f(t)) \in \calA(n,k)$ for all $t \in (-t_0,t_0)$, and $\varrho_{\eta}(0,f(0))=\varrho_{\star, \eta}$ is the minimizer of $S_\eta$ in $\calA(n,k)$, we have, since $F_\eta$ is continuously differentiable on $(-t_0,t_0)$ according to Lemma \ref{lem:EulerLag1}, 
\begin{equation*}
  \left.\frac{d F_{\eta}(t)}{dt}\right|_{t = 0} = 0.
\end{equation*}
With Lemma \ref{lem:EulerLag1}, this yields
\begin{equation*}
 \Tr \big(\varrho_{\star,\eta} \log(\varrho_{\star,\eta} + \eta) (\varrho_{1,\eta}+\varrho_{1,\eta}^*)\big)- \left\langle \frac{n[\varrho_{\star,\eta} \log(\varrho_{\star,\eta} + \eta)]}{n},n_{1,\eta} \right\rangle+\langle \gamma_{\eta},f'_3(0)\rangle = 0.
\end{equation*}
That \fref{eq:SmoothmEq} admits a unique solution in $L^\infty$ is a consequence of Lemma \ref{lem:Volterra} further since $K_\eta$ and $\gamma_\eta$ are bounded. Then, using \fref{sysf} in Proposition \ref{direct}, we simply remark that
\begin{align*}
\langle \gamma_{\eta},f_3'(0)\rangle&=-\langle m_{0,\eta},f_3'(0)-\calL_{K_{\eta}} f_3'(0)\rangle =-\langle m_{0,\eta},b\rangle
\\ & = -\left\langle m_{0,\eta}, \frac{n_{1,\eta}}{2n} - \frac{a}{2}\left(k_{1,\eta} - \frac{\nabla\sqrt{n}}{\sqrt{n}} \nabla n_{1,\eta}\right)\right\rangle
\\ & =\left\langle m_{\eta},k_{1,\eta}\right\rangle - \left\langle \frac{m_{\eta}}{n}(k-|\nabla\sqrt{n}|^2),n_{1,\eta} \right\rangle-\left\langle m_{\eta},\frac{\nabla\sqrt{n}}{\sqrt{n}}  \nabla n_{1,\eta}\right\rangle,
\end{align*}
where we recall that $a^{-1} = k - |\nabla\sqrt{n}|^2$. This ends the proof.
\end{proof}
\paragraph{Passing to the limit.}

We now pass to the limit $\eta \to 0$ in \fref{eq:EulerLagrangeAb} to obtain the Euler-Lagrange equation for $\varrho_{\star}$. We will need for this the following lemma, proved in Section \ref{proofconveta}.

\begin{lemma}\label{lem:conveta} 
Let $\{\eta_\ell\}_{\ell\in\mathbb{N}}\subset (0,1]$ be a sequence converging to $0$ and denote $\varrho_\ell = \varrho_{\star,\eta_\ell}$. Then:
\begin{enumerate}
\item $\{\varrho_\ell\}_{\ell\in\mathbb{N}}$ converges to $\varrho_{\star}$ in $\calJ_1$, and $\{\sqrt{\varrho_\ell}\}_{\ell\in\mathbb{N}}$ converges to $\sqrt{\varrho_{\star}}$ in $\calJ_2$.
  \item $\{\sqrt{H_0}\sqrt{\varrho_\ell}\}_{\ell\in\mathbb{N}}$ converges to $\sqrt{H_0}\sqrt{\varrho_{\star}}$ in $\calJ_2$.
\item for any $p\in\mathbb{N}$, the eigenvalues $\{\rho_{p,\eta_\ell}\}_{\ell\in\mathbb{N}}$ converge to $\rho_p$.
\item there exist a sequence of orthonormal eigenbasis $\{\phi_{p,\eta_\ell}\}_{p,\ell\in\mathbb{N}}$ of $\varrho_{\star,\eta_\ell}$, and an orthonormal eigenbasis $\{\phi_{p}\}_{p\in\mathbb{N}}$ of $\vs$ such that, for any $p\in\mathbb{N}$, 
\begin{itemize}
\item $\{\phi_{p,\eta_\ell}\}_{\ell\in\mathbb{N}}$ converges to $\phi_p$ in $L^2$,
\item $\{\sqrt{\rho_{p,\eta_\ell}}\,\nabla \phi_{p,\eta_\ell}\}_{\ell\in\mathbb{N}}$ converges to $\sqrt{\rho_p}\,\nabla \phi_p$ in $L^2$.
\end{itemize}
%\item $(s_{\eta_\ell}(\varrho_{\ell}))_{\ell\in\mathbb{N}}$ converge to $s(\varrho_{\star})$ in $\calJ_1$,
\item $\{\varrho_\ell\log(\varrho_\ell+\eta_\ell)\}_{\ell\in\mathbb{N}}$ converges to $\varrho_{\star}\log(\varrho_{\star})$ in $\calJ_1$.
  \item $\{\sqrt{\varrho_\ell}\log(\varrho_\ell+\eta_\ell)\}_{\ell\in\mathbb{N}}$ converges to $\sqrt{\varrho_{\star}}\log(\varrho_{\star})$ in $\calL(L^2)$.

\end{enumerate}
\end{lemma}
% \begin{proof}
% Since, $\forall k\in\mathbb{N}$, $\varrho_k\in\mathcal{E}^+$, $k[\varrho_{k}] = k$ and $n[\varrho_k] = n$, we deduce that $(\varrho_k)_{k\in\mathbb{N}}$ is a bounded sequence of $\mathcal{E}^+$. It follows from \cite[Lemma 3.1]{MP-Jstat} that there exists a converging subsequence (that we still denote $(\varrho_k)_{k\in\mathbb{N}}$) such that $\varrho_k \to \varrho$ in $\calE^+$. From here, we can follow the arguments from the proof of \cite[Lemma 3.7]{duboscq2019constrained} to deduce the desired results.
% \end{proof}

Following Lemma \ref{lem:conveta}, we suppose that the basis of eigenvectors that we were using for $\vs$ is the one from item (4). We are now in position to establish the next results. In the rest of the section, $\{\eta_\ell\}_{\ell\in\mathbb{N}}\subset (0,1]$ is a sequence such that $\eta_\ell\to 0$ as $\ell \to \infty$.

\begin{lemma} \label{convK} $K_{\eta_\ell}$ converges to $K$ in $L^2 \times L^2$ and $\gamma_{\eta_\ell}$ converges to $\gamma_\star$ in $L^2$, where $\gamma_\star$ is defined in \fref{defgamma}.
\end{lemma}
\begin{proof}
  Write
  $$
  K_{\eta}-K=G_{1,\eta}+G_{2,\eta},
  $$
  where
  \bee
  G_{1,\eta}(x,y)&=&\frac{1}{2 n(x)} \left( K_{0,\eta}(x,y)-K_{0}(x,y)\right)\\
  G_{2,\eta}(x,y)&=&\frac{a(x) \nabla \sqrt{n(x)}}{2\sqrt{n(x)}}\left( \nabla_x K_{0,\eta}(x,y)-\nabla_x K_{0}(x,y)\right).
  \eee
  Since $a, \sqrt{n}, n^{-1}$ are all bounded, we have the estimates
  $$
   \|G_{1,\eta}\|_{L^2 \times L^2} \lesssim \| K_{0,\eta}-K_{0} \|_{L^2 \times L^2}, \qquad \|G_{2,\eta}\|_{L^2 \times L^2} \lesssim \| \nabla_x K_{0,\eta}-\nabla_x K_{0} \|_{L^2 \times L^2}.
   $$
Writing $K_0=2 \Re \widetilde{K_0}$ and $K_{0,\eta}=2 \Re \widetilde{K_{0,\eta}}$(with obvious notation), and remarking that $\widetilde{K_{0}}$ and $\nabla \widetilde{K_{0}}$ (resp. $\widetilde{K_{0,\eta}}$ and $\nabla \widetilde{K_{0,\eta}})$ are the integral kernels of $-\overline{\vs \nabla}$ and $- \overline{\nabla \vs \nabla}$ (resp. $-\overline{\varrho_{\star,\eta} \nabla}$ and $- \overline{\nabla \varrho_{\star,\eta} \nabla}$), which are all trace class, we have
   $$
   \| \widetilde{K_{0,\eta_\ell}}-\widetilde{K_{0}} \|_{L^2 \times L^2} =\|\overline{\varrho_{\star,\eta_\ell} \nabla}- \overline{\varrho_{\star} \nabla}\|_{\calJ_2}, \qquad  \| \widetilde{K_{0,\eta_\ell}}-\widetilde{K_{0}} \|_{L^2 \times L^2} =\| \overline{\nabla \varrho_{\star,\eta_\ell} \nabla}- \overline{\nabla \varrho_{\star} \nabla}\|_{\calJ_2}.
   $$
   Since $\overline{\varrho_{\star,\eta} \nabla }=(\nabla \varrho_{\star,\eta})^*$, since moreover $ \sqrt{\varrho_{\star,\eta_\ell}}$ converges to $ \sqrt{\varrho_{\star}}$ in $\calJ_2$ according to Lemma \ref{lem:conveta} (1), and $\nabla \sqrt{\varrho_{\star,\eta_\ell}}$ converges to $\nabla \sqrt{\varrho_{\star}}$ in $\calJ_2$ according to Lemma \ref{lem:conveta} (2) (to see this, write e.g. $\nabla \sqrt{\varrho_{\star,\eta_\ell}}= \nabla (I+\sqrt{H_0})^{-1} (I+\sqrt{H_0}) \sqrt{\varrho_{\star,\eta_\ell}})$, both terms above converge to zero. This proves the first result on $K_{\eta_\ell}$. For the second one, we write
   $$
   \gamma_\eta= c_{1,\eta}+c_{2,\eta},
   $$
   where
  % $$
  % c_{1,\eta}= \nabla  \log (\varrho_{\star,\eta}) \varrho_{\star,\eta} \un_{(x,1)}, \qquad c_{2,\eta}= -\nabla  \varrho_{\star,\eta} \left(\un_{(x,1)} \frac{n [\varrho_{\star,\eta} \log(\varrho_{\star,\eta}+\eta)]}{n}\right)
  % $$
   $$
   c_{1,\eta}(x)=- 2 \Re \int_x^1 \calK_{1,\eta}(x,y) dy, \quad c_{2,\eta}(x)=- 2 \Re \int_x^1 \calK_{2,\eta}(x,y) \frac{n [\varrho_{\star,\eta} \log(\varrho_{\star,\eta}+\eta)]}{n}(y) dy
   $$
   with
   $$
      \calK_{1,\eta}(x,y)=  \sum_{ j \in \Nm} \rho_{j,\eta} \log(\rho_{j,\eta}+\eta) \nabla \phi_{j,\eta}(x) \phi_{j,\eta}^*(y), \qquad 
\calK_{2,\eta}(x,y)=  \sum_{ j \in \Nm} \rho_{j,\eta} \nabla \phi_{j,\eta}(x) \phi_{j,\eta}^*(y) .
$$
We remark that $\calK_{1,\eta}$ and $\calK_{2,\eta}$ are the integral kernels of $\nabla\varrho_{\star,\eta} \log(\varrho_{\star,\eta}+\eta)$ and $\nabla \varrho_{\star,\eta}$. These two operators are in $\calJ_2$ since $\nabla \varrho_{\star,\eta}\log(\varrho_{\star,\eta}+\eta)=\nabla \sqrt{\varrho_{\star,\eta}}\sqrt{\varrho_{\star,\eta}} \log(\varrho_{\star,\eta}+\eta)$, with $\nabla \sqrt{\varrho_{\star,\eta}} \in \calJ_2$, $\sqrt{\varrho_{\star,\eta}} \log(\varrho_{\star,\eta}+\eta)$ and $\sqrt{\varrho_{\star,\eta}}$ bounded. We already know that $\nabla \sqrt{\varrho_{\star,\eta_k}}$ converges to $\nabla \sqrt{\varrho_{\star}}$ in $\calJ_2$, which, together with Lemma \ref{lem:conveta} (6), shows that $\calK_{1,\eta_k}$ converges in $L^2\times L^2$ to the kernel $\calK_1$ of $\nabla\varrho_{\star} \log(\varrho_{\star})$. Hence,
$$\lim_{\ell \to \infty} c_{1,\eta_\ell} = c_1=- 2 \Re \int_x^1 \calK_{1}(x,y) dy, \qquad \textrm{in  } L^2. 
$$
In the same way, since $\calK_{2,\eta_k}$ converges in $L^2 \times L^2$ to $\calK_2$ the kernel of $\nabla \varrho_{\star}$,
$$
\lim_{\ell \to \infty} c_{2,\eta_\ell} =\lim_{\ell \to \infty} c_{2,\eta_\ell}':=-2 \lim_{\ell \to \infty} \Re \int_x^1 \calK_{2}(x,y) \frac{n [\varrho_{\star,\eta_\ell} \log(\varrho_{\star,\eta_\ell}+\eta)]}{n}(y) dy, \qquad \textrm{in  }L^2.
$$
To conclude, we deduce from Lemma \ref{lem:conveta} (5) and Lemma \ref{estimnlog} that $n[\varrho_{\star,\eta_\ell} \log(\varrho_{\star,\eta_\ell}+\eta_\ell)]$ converges weakly-$*$ in $L^\infty$ to $n[\varrho_{\star} \log(\varrho_{\star})]$. Since $\calK_2 \in L^2 \times L^2$, we can conclude that $c_{2,\eta_\ell}'$ converges strongly in $L^2$ to
$$
-2\Re \int_x^1 \calK_{2}(x,y) \frac{n [\varrho_{\star} \log(\varrho_{\star})]}{n}(y) dy.
$$
This ends the proof.
  \end{proof}

\bigskip

We have all needed now to pass to the limit in \fref{eq:EulerLagrangeAb} and obtain the Euler-Lagrange equation for $\vs$.

\begin{proposition}\label{prop:eulerL} (Euler-Lagrange equation). For $\gamma_\star \in L^\infty$ defined in \fref{defgamma},  the adjoint problem 
\begin{equation}\label{eq:adjm0}
m_{0}=\calL^*_{K} m_{0} -\gamma_\star,
\end{equation}
admits a unique solution in $L^\infty$. With $m_\star=a m_{0}/2$, we have the Euler-Lagrange equation
\begin{equation}\label{eq:EulerLagrangeA}
\Tr \big(\varrho_{\star} \log(\varrho_{\star})(\varrho_1+\varrho_1^*)\big)+\langle m_\star ,k_1\rangle+\langle A_\star,n_1\rangle + \left\langle m_\star \frac{|\nabla\sqrt{n}|^2}{n},n_1\right\rangle-\left\langle m_\star,\frac{\nabla \sqrt{n}}{\sqrt{n}} \nabla n_1\right\rangle=0,
\end{equation}
where
\begin{equation*}
A_\star = -\frac{n[\varrho_{\star} \log(\varrho_{\star})]}{n}-\frac{k}{n}m_\star,
\end{equation*}
and $n_1=n[\varrho_1 \varrho_{\star} +\varrho_{\star}\varrho_1^*]\in L^\infty$, $k_1=k[\varrho_1 \varrho_{\star} +\varrho_{\star}\varrho_1^*]\in L^\infty$.
\end{proposition}
\begin{proof}
We denote $\varrho_\ell = \varrho_{\star,\eta_\ell}$. We remark first that \eqref{eq:adjm0} has a unique solution $m_0 \in L^\infty$ according to Lemma \ref{lem:Volterra} since $K \in L^\infty \times L^\infty$ and $\gamma_\star \in L^\infty$. The first step of the proof consists in taking the limit of $m_{\eta_\ell}$, and the second one to pass to the limit in \eqref{eq:EulerLagrangeAb}. %We show below that \fref{eq:EulerLagrangeA} is obtained by passing to the limit in \eqref{eq:EulerLagrangeAb}.

\noindent\textbf{Step 1:} Consider $m_{0,\ell} = m_{0,\eta_\ell}$ which is the solution of \eqref{eq:SmoothmEq}. Taking the difference between  \eqref{eq:SmoothmEq} and  \eqref{eq:adjm0}, we have
\begin{equation*}
m_{0,\ell} - m_0 = \mathcal{L}^*_{K_{\eta_\ell}}(m_{0,\ell} - m_0) + \mathcal{L}^*_{K_{\eta_\ell} - K}m_0 + \gamma_{\eta_\ell} - \gamma_\star.
\end{equation*}
Since $K$ is bounded and $K_{\eta_\ell}$ verifies \fref{estKeta}, estimate \fref{estVolt} yields
$$
\|m_{0,\ell} - m_0\|_{L^2} \lesssim  \|\mathcal{L}^*_{K_{\eta_\ell} - K}m_0\|_{L^2}+\|\gamma_{\eta_\ell} - \gamma_\star\|_{L^2}.
$$
Since  $\|\mathcal{L}^*_{K_{\eta_\ell} - K}m_0\|_{L^2}\leq \|K_{\eta_\ell} - K\|_{L^2 \times L^2} \|m_0\|_{L^2}$, it follows from Lemma \ref{convK} that $m_{0,\ell}$ converges to $m$ in $L^2$.

\noindent\textbf{Step 2:} We pass now to the limit in \fref{eq:EulerLagrangeAb}. Recalling that $\varrho_{1,\eta_\ell}=\sqrt{\rho_{p,\eta_\ell}} \ket{\varphi} \bra{\phi_{p,\ell}}$, we conclude from Lemma \ref{lem:conveta} (3) and (4) that $\varrho_{1,\eta_\ell}$ converges to $\varrho_1=\sqrt{\rho_{p}} \ket{\varphi} \bra{\phi_{p}}$ in $\calL(L^2)$. Lemma \ref{lem:conveta} (5) then yields
\begin{align*}
\lim_{\ell \to \infty} \Tr \big(\varrho_{\ell} \log(\varrho_{\ell}+\eta_\ell)\varrho_{1,\ell}\big)=\Tr \big(\varrho_{\star} \log(\varrho_{\star})\varrho_1\big).
\end{align*}
 Moreover, Lemma \ref{lem:conveta} (5) shows that $n[\varrho_{k} \log(\varrho_{k}+\eta_k)]$ converges to $n[\varrho\log(\varrho)]$ in $L^1$, and because of Lemma \ref{estimnlog}, the convergence holds also weakly-$*$ in $L^\infty$. This shows that $A_{\eta_k}$ converges to $A$ weakly-$*$ in $L^\infty$. Finally, with
$$
n_{1,\eta}= 2 \rho_{p,\eta}^{3/2}\,\Re \,\varphi^* \phi_{p,\eta}, \qquad k_{1,\eta}=2 \rho_{p,\eta}^{3/2}\,\Re \,\nabla \varphi^* \nabla \phi_{p,\eta}, 
$$
and Lemma \ref{lem:conveta} (3)-(4), it follows that $k_{1,\eta_\ell}$ converges to $k_1$ in $L^2$ and that $n_{1,\eta_\ell}$ converges to $n_1$ in $H^1$. With the fact that $\nabla \sqrt{n} /n$ is bounded, we have therefore sufficient compactness to pass to the limit in \fref{eq:EulerLagrangeAb} to recover \fref{eq:EulerLagrangeA}. This ends the proof.
\end{proof}

\begin{remark} \label{rem3}A by-product of the proof of Proposition \ref{prop:eulerL} is that that $m_{\eta_\ell}$ converges to $m_\star$ in $L^2$, and that $A_{\eta_\ell}$ converges to $A_\star$ in $L^1$ and weakly-$*$ in $L^\infty$. We will use these facts further.
\end{remark}
We prove in the next section that $\vs$ is full rank.
\subsection{Step 3: The minimizer is full rank}
Let $\psi, \varphi \in H^1$, and define
%\begin{equation}\label{eq:SesquilinearF}
$$
\mathcal{Q}_{\star,\eta}(\psi,\varphi) = \int_0^1 n(x) \left(\nabla \left( \frac{\psi^*(x)}{\sqrt{n(x)}}\right) \nabla \left( \frac{\varphi(x)}{\sqrt{n(x)}}\right)\right)m_\eta(x)  dx + \int_0^1 A_{\eta}(x)  \psi^*(x) \varphi(x) dx.
$$
%\end{equation}
The fact that
\begin{align*}
n(x)\left(\nabla \left( \frac{\psi^*(x)}{\sqrt{n(x)}}\right) \cdot \nabla \left( \frac{\varphi(x)}{\sqrt{n(x)}}\right)\right) &=  \nabla  \psi^*(x)  \nabla \varphi(x)  +  \left|\frac{\nabla \sqrt{n(x)}}{\sqrt{n(x)}}\right|^2 \psi^*(x)\varphi(x)
 \\ &\hspace{1em} - \frac{\nabla\sqrt{n(x)}}{\sqrt{n(x)}}\cdot \nabla\left(\psi^*(x)\varphi(x)  \right),
\end{align*}
yields
$$
\langle m_\eta ,\nabla \psi^* \nabla \varphi\rangle+\langle A_\eta,\psi^* \varphi\rangle+\left\langle m_{\eta} \frac{|\nabla\sqrt{n}|^2}{n},\psi^* \varphi\right\rangle-\left\langle m_\eta,\frac{\nabla \sqrt{n}}{\sqrt{n}} \nabla (\psi^* \varphi)\right\rangle=\mathcal{Q}_{\star,\eta}(\psi,\varphi).
$$
% \begin{remark}
% Since $m\in L^{\infty}$, we have, for any $\psi,\varphi\in H^{1}$, by a Sobolev embedding,
% \begin{align*}
% \left| \mathcal{Q}_m(\psi,\varphi)\right| &\lesssim \|\nabla\psi\|_{L^2}\|\nabla \varphi\|_{L^2}\|m\|_{L^{\infty}}
% \\ &\hspace{1em}+ (n_m)^{-1}  \|n[\varrho_\star \log \varrho_\star]\|_{L^2} \|\psi\|_{H^1} \|\varphi\|_{L^2}
% \\ &\hspace{1em}+ \left((n_m)^{-1} \|k\|_{L^{\infty}}+ \frac{\|\nabla \sqrt{n}\|_{L^{\infty}}^2}{n_m}\right)\|\psi\|_{L^2} \|\varphi\|_{L^2}\|m\|_{L^{\infty}}
% \\ &\hspace{1em}+ \|m\|_{L^{\infty}}\frac{\|\nabla \sqrt{n}\|_{L^{\infty}}}{\sqrt{n_m}} \left(\|\nabla \psi\|_{L^2} \|\varphi\|_{L^2}+ \| \psi\|_{L^2} \|\nabla\varphi\|_{L^2}\right)<+\infty.
% \end{align*}
% \end{remark}

The following result is central in proving the full rank character.
\begin{proposition}\label{prop:ineqQ}
Let $\psi\in W^{1,\infty}$ and consider the rank one operator
\begin{equation*}
P = \ket{\psi} \bra{\psi}.
\end{equation*}
Then, for any $\eta\in(0,1/2)$, the following inequality holds
\begin{equation}\label{eq:fullrnkContrad}
\Tr \big(\log(\varrho_{\star,\eta} + \eta) P \big) + \mathcal{Q}_{\star,\eta}(\psi,\psi) \geq  0.
\end{equation}
%where
% \begin{equation*}
% \mathcal{Q}_{\star,\eta}(\psi,\psi)=\langle m_\eta,|\nabla \psi|^2\rangle+\langle A_\eta,|\psi|^2\rangle+ \left\langle m_{\eta},\frac{|\nabla\sqrt{n}|^2}{n}|\psi|^2\right\rangle-\left\langle m_{\eta},\frac{\nabla \sqrt{n}}{\sqrt{n}} \nabla |\psi|^2\right\rangle.
% \end{equation*}
\end{proposition}
\begin{proof}
The proof is almost identical to \fref{eq:EulerLagrangeAb}, with the following differences: we set $\varrho_1=0$, and $\varrho_2=P$; with such a choice $\varrho_\eta(t,f(t))$ belongs to the feasible set $\calA(n,k)$ for for positive $t \in [0,t_0)$ only. We therefore have now the inequality
\begin{equation*}
\lim_{t\to0^+}  \frac{F_{\eta}(t) - F_{\eta}(0)}{t} \geq 0,
\end{equation*}
and not an equality. Replacing then $\varrho_{1,\eta}+\varrho^*_{1,\eta}$ by $P$, $n_{1,\eta}$ by $|\psi|^2$, and $k_{1,\eta}$ by $|\nabla \psi|^2$ in \fref{eq:EulerLagrangeAb}, and the equality by an inequality, we obtain \fref{eq:fullrnkContrad}.
\end{proof}

\bigskip

We can now prove that the minimizer is full rank.

\begin{proposition}\label{prop:fullrank}
The kernel of the minimizer $\varrho_{\star}$ is $\{0\}$.
\end{proposition}
\begin{proof}
The proof is based on a contradiction argument, as in \cite[Section 5]{MP-JSP}, by differentiating in a direction related to a nonzero eigenfunction in the kernel of $\varrho_\star$.

\noindent\textbf{Step 1:} We assume that the kernel of $\varrho_\star$ is not $\{0\}$, and consider an orthonormal basis $\{\psi_p\}_{p\in I}$ of $\mbox{Ker}\,(\varrho_\star)$ ($I$ may be empty, finite or infinite, and we write $|I|$ for its cardinal). Then, we denote by $\{\rho_j\}_{1\leq j\leq N}$ the nonincreasing sequence of nonzero eigenvalues of $\varrho_\star$ (here $N$ is finite or not), associated to the orthonormal family of eigenfunctions $\{\phi_j\}_{1\leq j\leq N}$. We thus obtain a Hilbert basis $\{\{\psi_j\}_{1 \leq j \leq |I|}, \{\phi_j\}_{1\leq j\leq  N}\}$ of $L^2$. Pick then for instance $\psi_1$, that we denote for simplicity by $\psi$. At this point, we only know that $\psi$ belongs in $L^2$, which is not sufficient for our purpose. We regularize it by letting
\begin{equation*}
\psi_{\varepsilon} = (1 + \varepsilon H_0)^{-1} \psi,
\end{equation*}
for any $\varepsilon>0$. It follows that $\psi_{\varepsilon}\in H^2\subset W^{1,\infty}$, and $\psi_{\varepsilon}\to\psi$ in $L^2$.

\noindent\textbf{Step 2:}  We now consider a sequence $\{\eta_\ell\}_{\ell\in\mathbb{N}}\subset (0,1/2)$ such that $\eta_\ell \to 0$ as $\ell\to+\infty$. By using Lemma \ref{lem:conveta}, we know that there exists  a sequence of eigenfunctions $\{\psi_{j,\ell}\}_{1\leq j\leq |I|,\ell\in\mathbb{N}}$ of $\{\varrho_{\star,\eta_\ell}\}_{\ell\in\mathbb{N}}$, associated to a sequence of eigenvalues $\{\lambda_{j,\ell}\}_{1\leq j\leq |I|,\ell\in\mathbb{N}}$, such that, for any $1\leq j\leq |I|$,
\begin{equation*}
\psi_{j,\ell} \underset{\ell\to\infty}\to \psi_j\quad\textrm{and}\quad \lambda_{j,\ell}\underset{\ell\to\infty}\to 0.
\end{equation*}
For clarity, we simply denote $\psi_\ell = \psi_{1,\ell}$ and $\lambda_\ell = \lambda_{1,\ell}$. It can be shown with Remark \ref{rem3} that $\mathcal{Q}_{\star,\eta_\ell}(\psi_{\varepsilon},\psi_{\varepsilon}) \to \mathcal{Q}_{\star}(\psi_{\varepsilon},\psi_{\varepsilon})$ when $\ell\to\infty$, and as a consequence, there exists $\ell_1(\eps)\in\mathbb{N}$ such that
\begin{equation*}
\mathcal{Q}_{\star,\eta_\ell}(\psi_{\varepsilon},\psi_{\varepsilon}) \leq |\mathcal{Q}_{\star}(\psi_{\varepsilon},\psi_{\varepsilon})| + 1,
\end{equation*}
for any $\ell\geq \ell_1(\eps)$.
Furthermore, for $P_\eps=\ket{\psi_\eps} \bra{\psi_\eps}$, we have
\begin{align*}
\Tr \big(\log(\varrho_{\star,\eta_\ell} + \eta_\ell) P_\eps \big) &= \log(\lambda_{\ell} + \eta_\ell)|\langle \psi_{\varepsilon},\psi_{\ell}\rangle|^2 + \sum_{j = 2}^{|I|}  \log(\lambda_{j,\ell}  + \eta_\ell) |\langle \psi_{\varepsilon},\psi_{j,\eta_\ell}\rangle|^2
\\ &\hspace{1em} + \sum_{j = 1}^N \log(\rho_{j,\eta_\ell} + \eta_\ell) |\langle \psi_{\varepsilon},\phi_{j,\eta_\ell}\rangle|^2,
\end{align*}
and there exists a $\ell_2\geq \ell_1(\eps)$ such that $\log(\lambda_{\ell} + \eta_\ell)\leq 0$ and $\log(\rho_{N_0,\eta_\ell} + \eta_\ell)\leq 0$ for any $\ell\geq \ell_2$ and a certain $N_0\in\mathbb{N}$. This yields, for all $\ell\geq \ell_2$,
\begin{align*}
\Tr \big(\log(\varrho_{\star,\eta_\ell} + \eta_\ell) P_\eps \big) &\leq \log(\lambda_{\ell} + \eta_\ell)|\langle \psi_{\varepsilon},\psi_{\ell}\rangle|^2 + \sum_{j = 1}^{N_0} \log(\rho_{j,\eta_\ell} + \eta_\ell) |\langle \psi_{\varepsilon},\phi_{j,\eta_\ell}\rangle|^2.
\end{align*}
Remarking that
\begin{equation*}
|\langle \psi_{\varepsilon},\psi\rangle| \geq \|\psi\|_{L^2}^2 - |\langle \psi_{\varepsilon}-\psi,\psi\rangle| \geq 1 - \|\psi_{\varepsilon}-\psi\|_{L^2},
\end{equation*}
there exists $\varepsilon_0\in (0,1)$ such that $|\langle \psi_{\varepsilon},\psi\rangle|\geq 1/2$ for all $\varepsilon\in (0,\varepsilon_0)$. Moreover, since $\psi_{\ell}\to\psi$ as $\ell\to\infty$ in $L^2$ and $\|\psi_\eps\|_{L^2} \leq \|\psi\|_{L^2}$, there exists a $\ell_3\geq \ell_2$ such that $|\langle \psi_{\varepsilon},\psi_{\ell} \rangle|\geq 1/4$ for any $\varepsilon\in(0,\varepsilon_0)$ and any $\ell\geq \ell_3$. Hence, for any $\varepsilon\in(0,\varepsilon_0)$ and $\ell\geq \ell_3$, we obtain that
\begin{align*}
\Tr \big(\log(\varrho_{\star,\eta_\ell} + \eta_\ell) P_\eps \big) &\leq \frac{1}{16} \log(\lambda_{\ell} + \eta_\ell) + \log(M+1/2)\sum_{j = 1}^{N_0} |\langle \psi_{\varepsilon},\phi_{j,\eta_\ell}\rangle|^2,
\end{align*}
where $M=\Tr(\varrho_{\star,\eta_\ell})=\|n\|_{L^1}$ is such that $\rho_{j,\eta_\ell} \leq M$ for all $j$ and $\ell$. Thus, we deduce the following inequality, for any $\varepsilon\in(0,\varepsilon_0)$ and $\ell\geq \ell_3$,
\begin{align*}
\Tr \big(\log(\varrho_{\star,\eta_\ell} + \eta_\ell) P_\eps \big) + \mathcal{Q}_{\star,\eta_\ell}(\psi_{\varepsilon},\psi_{\varepsilon}) &\leq \frac{1}{16} \log(\lambda_\ell + \eta_\ell) + \log(M+1/2) \|\psi_\eps\|^2_{L^2}
\\ &\hspace{1em} + |\mathcal{Q}_{m}(\psi_{\varepsilon},\psi_{\varepsilon})| + 1.
\end{align*}
 This contradicts \eqref{eq:fullrnkContrad} by taking $\ell$ sufficiently large to make the r.h.s of the previous inequality negative. This ends the proof.
\end{proof}

\bigskip 
Let

$$
G_\star(\varphi)=-\sum_{j \in \mathbb{N}} \log (\rho_j)|(\phi_j,\varphi)|^2,
$$
where $\rho_j$ and $\phi_j$ are the eigenvalues and eigenfunctions of $\vs$. Note that $\log(\rho_j)$ is well-defined according to the previous proposition. Sending $\eta$ to zero in Proposition \ref{prop:ineqQ} gives the result below.

\begin{corollary}\label{prop:QineqNN} Let $\varphi\in H^1$. Then, we have
\begin{equation}\label{eq:exprQm}
G_\star(\varphi) \leq \mathcal{Q}_\star(\varphi,\varphi).
\end{equation}
\end{corollary}

\begin{proof} We start from \fref{eq:fullrnkContrad} and need to regularize $\varphi$. Let then $\varphi_\eps=(I+\eps H_0)^{-1} \varphi \in H^2 \subset W^{1,\infty}$. % For $p \in \Nm$, let then
% $$
% \varphi_p(x)=A_p+\int_0^x \min\{\nabla \varphi(y),p\} dy,
% $$
% where
% $$
% A_p=\int_0^1 \varphi(y)dy-\int_0^1\int_0^x \min\{\nabla \varphi(y),p\} dydx.
% $$
% The function $\varphi_p$ belongs to $W^{1,\infty}$, and verifies
% \be
% \nabla \varphi_p = \min\{\nabla \varphi,p\} \underset{p \to \infty}\to \nabla \varphi \quad \textrm{a.e., \quad and} \qquad \|\varphi_p -\varphi\|_{L^2} \lesssim \| \nabla (\varphi_p-\varphi) \|_{L^2} \underset{p \to \infty}\to 0,
% \ee
% where we used the Poincar\'e-Wirtinger inequality and dominated convergence with $|\nabla \varphi_p| \leq |\nabla \varphi |$. 
With $P_\eps=\ket{\varphi_\eps} \bra{\varphi_\eps}$, we have
$$
\Tr \big(\log(\varrho_{\star,\eta_\ell} + \eta) P_\eps \big)= \sum_{j = 1 }^{+\infty}\log(\rho_{j,\eta_\ell} + \eta_\ell) |\langle\phi_{j,\eta_\ell},\varphi_\eps \rangle|^2.
$$
Let $N_0 = \argmin \{j\in\mathbb{N};\; -\log(\rho_{j} + 1/2) > 0 \}$. Since $\rho_{N_0,\eta_\ell}$ converges to $\rho_{N_0}$ according to Lemma \ref{lem:conveta} (3), there is an $\ell_0$ such that $-\log(\rho_{N_0,\eta_\ell} + 1/2) > 0$ for $\ell \geq \ell_0$. Then, using Fatou's lemma and Lemma \ref{lem:conveta} (3)-(4), we  obtain that
\begin{equation*}
\liminf_{\ell \to \infty} \sum_{j = N_{0} }^{+\infty}-\log(\rho_{j,\eta_\ell} + \eta) |(\phi_{j,\eta_\ell},\varphi_\eps)|^2 \geq \sum_{j = N_{0} }^{+\infty}-\log(\rho_{j}) |(\phi_{j},\varphi_\eps)|^2.
\end{equation*}
Moreover, Remark \ref{rem3} implies that
\begin{equation*}
\mathcal{Q}_{\star,\eta_\ell}(\varphi_\eps,\varphi_\eps) \underset{\ell \to \infty}\to\mathcal{Q}_{\star}(\varphi_\eps,\varphi_\eps).
\end{equation*}
Thus, passing to the limit $\ell \to \infty$ in \eqref{eq:fullrnkContrad}, we obtain the inequality
\begin{equation*}
\mathcal{Q}_{\star}(\varphi_\eps,\varphi_\eps) \geq - \sum_{j\in\mathbb{N}} \log(\rho_j) |(\phi_j,\varphi_\eps)|^2.
\end{equation*}
Since $\varphi_\eps \to \varphi$ in $H^1$ as $\eps \to 0$, and since $m_\star$, $A_\star$ and $|\nabla \sqrt{n}|/\sqrt{n}$ are all bounded, we can pass to the limit in the l.h.s above. Fatou's lemma finally allows us to pass to the limit in the r.h.s, which concludes the proof.
\end{proof}
\subsection{Step 4: Conclusion} \label{sec:conc}

To conclude the proof of Theorem \ref{mainth}, it remains to obtain the minimization principle \fref{minieig}, to prove that $u[\vs]=0$, and that $m_\star$ is nonnegative. The latter is addressed in the corollary below, and is a consequence of Proposition \ref{prop:QineqNN}.

\begin{corollary} %\label{cor:mpos}
 $m_\star$ is nonnegative a.e. on $[0,1]$.
\end{corollary}
\begin{proof} We proceed by contradiction. Let
  $$
  S=\{ x \in (0,1): m_\star(x)<0\; \; a.e.\},
  $$
and suppose the Lebesgue measure of $S$, denoted $|S|$, is not zero. Since
$$
S=\bigcup_{\ell=1}^\infty \left\{ x \in (0,1): -m_\star(x) \geq \ell^{-1}\; \; a.e.\right\}:=\bigcup_{\ell=1}^\infty S_\ell,
$$
there exists an $\ell_0$ such that $|S_{\ell_0}|>0$. Since $m_\star$ is only in $L^\infty$, we cannot conclude that $S_{\ell_0}$ is open. Nevertheless, by outer regularity of the Lebesgue measure, there exists, for any $\delta>0$, an open set $I$ such that $S_{\ell_0} \subset I$ and $|S_{\ell_0}| \geq |I|-\delta$. Let then $B(\eps_0)=(x_0-\eps_0,x_0+\eps_0) \subset I$. Consider $\varphi\in C^{\infty}(\Rm)$ with support in $[-1,1]$, and for $\eps \in (0,\eps_0)$, introduce 
\begin{equation*}
\varphi_\eps(x)=\eps^{-1/2}\varphi\left(\frac{x-x_0}{\eps}\right).
\end{equation*}
We remark that $\|\varphi_{\varepsilon}\|_{L^2} = \|\varphi\|_{L^2} $. 
Since $\varphi_\eps \in H^1$ and $\{\rho_{j}\}_{j\in\mathbb{N}}$ is nonincreasing, we have, from \eqref{eq:exprQm} and the fact that $n$ is bounded above and below, 
\begin{equation}\label{eq:mnonnegQbnd}
\gamma \leq - \log (\rho_1) \| \sqrt{n}\varphi_\eps\|^2 \leq \calQ_\star(\sqrt{n} \varphi_\eps,\sqrt{n} \varphi_\eps),
\end{equation}
where $\gamma$ is finite, positive or negative depending on whether $\rho_1<1$ or not.
We now split $B(\eps) \subset B(\eps_0) \subset I$ into $B(\eps)=B(\eps) \cap S_{\ell_0}+R$, where $R=\{x \in B(\eps), x \notin S_{\ell_0}\}$. By construction, $|R| \leq \delta$. Write then
\bee
 \int_0^1n|\nabla \varphi_\eps|^2 m_\star dx &=&  \int_{B(\eps) \cap S_{n_0} }n|\nabla \varphi_\eps|^2 m_\star dx+\int_{R}n|\nabla \varphi_\eps|^2 m_\star dx\\
&=&T_1+T_2.
\eee
We have
$$
|T_2| \leq \eps^{-2} \|n m_\star\|_{L^\infty} \|\nabla \varphi\|^2_{L^\infty} |R| \leq C_1 \delta \eps^{-2}.
$$
Choosing $\delta=\eps^2$, we find, for some $C_2>0$
\begin{align*}
  \calQ_\star(\sqrt{n} \varphi_\eps,\sqrt{n} \varphi_\eps)& \leq T_1+C_1  +\|\varphi\|_{L^2}^2  \|A_\star\|_{L^\infty} \leq T_1+ C_2.
\end{align*}
Since
\begin{equation*}
 T_1 \leq - \frac{n_m}{\ell_0 \eps^2} \|\nabla \varphi\|_{L^2}^2 = -C_4 \eps^{-2}, \qquad C_4>0,
\end{equation*}
where $n_m=\min_{x\in [0,1]} n(x)$, there exists an $\varepsilon>0$ sufficient small such that $\calQ_\star(\sqrt{n} \varphi_\eps,\sqrt{n} \varphi_\eps)$ is less than $\gamma$, which contradicts \fref{eq:mnonnegQbnd}. Hence $S$ is of measure zero and the proof is ended. \end{proof}

\bigskip

The first step toward the minimization principle is the next lemma. 

\begin{lemma}\label{lem:SesquOrth}
Let $\varphi\in H^1$. We have
\begin{equation}\label{eq:lemQ2}
 \calQ_\star(\phi_j,\varphi)= - \log(\rho_j) \langle \phi_j,\varphi\rangle, \quad\forall j\in\mathbb{N}.
\end{equation}
\end{lemma}
\begin{proof}
We start from \fref{eq:EulerLagrangeA}, and set $\varphi_{\varepsilon} = (1+\varepsilon H_0)^{-1} \varphi \in H^2 \subset W^{1,\infty}$. With $\varrho_1= \sqrt{\rho_j}\ket{\varphi_{\varepsilon}}\bra{\phi_j}$, we find
\begin{align*}
n_1 &= n[\varrho_1\varrho_{\star} + \varrho_{\star} \varrho_1^*]= \rho^{3/2}_j (\phi_j^* \varphi_{\varepsilon} + \varphi_{\varepsilon}^*\phi_j) = 2 \rho^{3/2}_j \Re(\phi_j^* \varphi_{\varepsilon})\in W^{1,\infty},
\end{align*}
and
\begin{align*}
k_1 &= k [\varrho_1\varrho_{\star} + \varrho_{\star} \varrho_1^*]= - n [\nabla (\varrho_1\varrho_{\star} + \varrho_{\star} \varrho_1^*)\nabla ]= 2 \rho^{3/2}_j \Re(\nabla \phi_j^* \nabla \varphi_{\varepsilon}) \in L^{\infty}.
\end{align*}
Hence, by using \eqref{eq:EulerLagrangeA} and the fact that $\rho_j>0$ for all $j$ according to Proposition \ref{prop:fullrank}, we can see that
\begin{equation*}
  \log(\rho_j) \Re \int_0^1 \phi_j^*(x) \varphi_{\varepsilon}(x) dx +   \Re \left(\calQ_\star(\phi_j,\varphi_{\varepsilon}) \right)= 0.
\end{equation*}
Following the same steps with $\varrho_1= i \sqrt{\rho_j}\ket{\varphi_{\varepsilon}}\bra{\phi_j}$, we obtain
\begin{equation*}
 \log(\varrho_j) \Im \int_0^1 \phi_j^*(x) \varphi_{\varepsilon}(x) dx +  \Im \left(\calQ_\star(\phi_j,\varphi_{\varepsilon}) \right)= 0,
\end{equation*}
which gives, by adding both expressions,
%\begin{equation}\label{eq:lemQ2reg}
$$ 
\calQ_\star(\phi_j,\varphi_{\varepsilon})= - \log(\rho_j) \langle \phi_j,\varphi_{\varepsilon}\rangle, \quad\forall j\in\mathbb{N}.
$$
%\end{equation}
Since $\varphi_\eps \to \varphi$ in $H^1$, and since $m_\star$, $A_\star$ and $|\nabla \sqrt{n}|/\sqrt{n}$ are all bounded, we can pass to the limit in the equation above and conclude the proof.
\end{proof}

\bigskip

With Lemma \ref{lem:SesquOrth}, we can now  prove that $u[\varrho_\star]=0$. Choosing indeed $\varphi=\phi_j \psi$ in \fref{eq:lemQ2} for $\psi$ real-valued and smooth, and taking the imaginary part, we find

$$
0=\Im \, Q_\star(\phi_j,\phi_j \psi)=\int_0^1 \Im \nabla \phi_j^* \phi_j \psi dx.
$$
Since $\psi$ is arbitrary, this implies that $\Im \nabla \phi_j^* \phi_j=0$ and therefore that $u[\vs]=0$ according to \fref{defcurrent}.

Regarding \fref{minieig}, we find from \fref{eq:exprQm}, since $\{\rho_j\}_{j\in\mathbb{N}}$ is a nonincreasing sequence, 
$$
- \log (\rho_1) \leq \inf_{\varphi \in H^1, \|\varphi\|_{L^2}=1} Q_\star(\varphi,\varphi).
$$
According to Lemma \ref{lem:SesquOrth}, we have $Q_\star(\phi_1,\phi_1)=- \log (\rho_1)$, and therefore the above infimum is attained at $\phi_1$. At any order $p>1$, we have, for any $\varphi \in \calK_p$, 
$$
- \log (\rho_p)\|\varphi\|_{L^2}^2 \leq G_\star(\varphi) \leq \calQ_\star(\varphi,\varphi) \qquad \textrm{so that} \qquad
- \log (\rho_p) \leq \inf_{\varphi \in \calK_p} Q_\star(\varphi,\varphi),
$$
and, according to Lemma \ref{lem:SesquOrth}, the infimum is attained at $\phi_p$. This proves \fref{minieig}, and concludes the proof of Theorem \ref{mainth}.

\section{Other proofs} \label{otherproof}

\subsection{Proof of Proposition \ref{direct}} \label{proofpropdire}

The proof is based on the implicit function theorem, and the first part consists in establishing some regularity for the perturbation $\varrho(t,f)$. In the entire proof, $\{\rho_j\}_{j \in \Nm}$ and $\{\phi_j\}_{j \in \Nm}$ are the eigenvalues and eigenfunctions of $\vs$.

\begin{lemma} \label{rhotE}Let $\varrho_1=\sqrt{\rho_p} \ket{\varphi}\bra{\phi_p}$ for $\varphi \in W^{1,\infty}$ and some $p \in \Nm$, let $\varrho_2\in\mathcal{E}^+$, and let $f=(f_2,f_3) \in  W^{1,\infty} \times L^{\infty}$. Then $\varrho(t,f) \in \calE$ for all $t \in [-1,1]$. %Then, $\varrho(t,f)$ is positive for all $t \in [-1,1]$ when $\varrho_2=0$, and positive for all $t \in [0,1]$ when $\varrho_2 \neq 0$. Moreover, $\varrho(t,f) \in \calE$ for all $t \in [-1,1]$. 
\end{lemma}

\begin{proof} %Set $\sigma=\vs+t \varrho_2$, which is positive in the cases stated in the lemma. The positivity of $\varrho(t,f)$ has already been addressed. 
Since $\calE$ is a Banach space, it suffices to show that each term in \fref{decomp} belongs to $\calE$, and since $\vs$ and $\varrho_2$ have identical roles, we only treat the term in $\vs$ and set $\varrho_2=0$. We need first to properly define $\sqrt{H_0} \varrho(t,f) \sqrt{H_0}$. This is done in the spirit of Lemma \ref{proper}: first, we remark that $\nabla \varrho_1$ is bounded since $\varphi \in W^{1,\infty}$; this shows that $\varrho_1: L^2 \to H^1$, and subsequently that $L(t,f):H^1 \to H^1$ since $f_2 \in W^{1,\infty}$. Together with $\sqrt{\vs} L^2 \subset H^1$ since $\vs \in \calE^+$, and with $D(\sqrt{H_0})=H^1$, this shows that $B=\sqrt{H_0} (I+L(t,f)) \sqrt{\vs}$ is bounded, and therefore that $\sqrt{H_0} \varrho(t,f) \sqrt{H_0}$ is interpreted as $B B^*$ as in Lemma \ref{proper}.

  We now prove that $\varrho(t,f) \in \calE$, that is $(I+L(t,f)) \sqrt{\vs} \in \calJ_2$ and $B \in \calJ_2$. According to \cite[Theorem 6.22,
item (g)]{RS-80-I}, it suffices to show that there are orthonormal basis $\{\psi_j\}_{j \in \Nm}$ and $\{e_j\}_{j \in \Nm}$ of $L^2$  such that  
$$
\sum_{j\in\mathbb{N}} \|(I+L(t,f))\sqrt{\vs}e_j\|_{L^2}^2 < \infty \qquad \textrm{and} \qquad \sum_{j\in\mathbb{N}} \|\sqrt{H_0} (I+L(t,f))\sqrt{\vs}\psi_j\|_{L^2}^2 <\infty.
$$
Setting $e_j=\psi_j=\phi_j$ , we find
\begin{align} 
&(I+L(t,f)) \sqrt{\vs} \phi_j=\sqrt{\rho_j} (\phi_j+t \varrho_1 \phi_j+ f_2 \phi_j+ T_{f_3} \phi_j) \label{eL}\\
&\nabla(I+L(t,f))\sqrt{\vs} \phi_j = \sqrt{\rho_j} (\nabla \phi_j +t \nabla \varrho_1 \phi_j + (f_2 + f_3 )\nabla \phi_j+ \nabla f_2 \phi_j). \label{nabL}
\end{align}
This leads to, for all $t \in [-1,1]$,
\begin{align*}
  \sum_{j\in\mathbb{N}} \|(1+L(t,f))\sqrt{\vs}\phi_j\|_{L^2}^2 &\leq  (1+\| \varrho_1\|_{\calL(L^2)}+\|f_2\|_{L^{\infty}} + \|f_3\|_{L^{\infty}})^2\sum_{j\in\mathbb{N}} \rho_j \|\phi_j\|_{H^1}^2 \\ & \lesssim  \| \vs \|_{\calE}<+\infty,
\end{align*}
showing that $\varrho(t,f)$ is trace class. Regarding the bound in $\calE$, we remark that $\|\sqrt{H_0} \varphi\|_{L^2}=\|\nabla \varphi\|_{L^2}$ for all $\varphi \in H^1$, and that, with \fref{nabL},
\begin{align*}%\label{eq:estmGFrechet1}
\|\nabla(I+L(t,f))\phi_j\|_{L^2} &\leq (1+\|f_2\|_{L^\infty}+\|f_3\|_{L^\infty} )\|\nabla \phi_j\|_{L^2} + (\|\nabla\varrho_1\|_{\calL(L^2)}+\|\nabla f_2\|_{L^\infty}) \|\phi_j\|_{L^2}\nonumber
\\ &\lesssim (1+\|f_2\|_{W^{1,\infty}}+\|f_3\|_{L^{\infty}}  + \|\nabla\varrho_1\|_{\calL(L^2)} ) \|\phi_j\|_{H^1}.
\end{align*}
Since we have seen above that $\nabla \varrho_1$ is bounded, we deduce, for all $t \in [-1,1]$,
\begin{equation*}
\sum_{j\in\mathbb{N}} \|\sqrt{H_0} (1+L(t,f))\sqrt{\vs}\phi_j\|_{L^2}^2 \lesssim \sum_{j\in\mathbb{N}} \rho_j \|\phi_j\|_{H^1}^2 <+\infty,
\end{equation*}
leading to $\varrho(t,f)\in\calE$. This ends the proof.
\end{proof}

\bigskip 

We now consider the following function
$$
G(t,f) = \begin{pmatrix} G_1(t,f) \\ G_2(t,f)\end{pmatrix}=  \begin{pmatrix} n[\varrho(t,f)]-n \\ k[\varrho(t,f)]-k\end{pmatrix}.
$$
Note that $G$ is well-defined since $\varrho(t,f) \in \calE$ according to the preceding Lemma, and that $G$ is real-valued since $\varrho(t,f)$ is self-adjoint. In the next lemma, $D_f G$ denotes the differential of $G$ w.r.t $f$.

\begin{lemma}\label{lemiso} Let $\varrho_1=\sqrt{\rho_p} \ket{\varphi}\bra{\phi_p}$ for $\varphi \in W^{1,\infty}$ and some $p \in \Nm$, and let $\varrho_2\in\mathcal{E}^+$ with $k[\varrho_2]\in L^{\infty}$. Then:\\
  \noindent (i) $G$ is continuously Fr\'echet differentiable from $[-1,1]\times  L^{\infty}\times W^{1,\infty}$ to $L^{\infty}\times W^{1,\infty}$.\\
  \noindent (ii) $D_fG(0,0)$ is an isomorphism from $W_r^{1,\infty} \times L^{\infty}_r$ to itself, where $W_r^{1,\infty}$ (resp. $L^{\infty}_r$) is the space of real $W^{1,\infty}$ (resp. $L^\infty$) functions.  
\end{lemma}

The proof of Lemma \ref{lemiso} is somewhat long and given further. We recall below the implicit function theorem on Banach spaces.

\begin{theorem} \label{thimp}
  Let $X,Y, Z$ be three Banach spaces, $O$ an open set of $X \times Y$,  and $G: O \to Z$ be a continuously Fr\'echet differentiable mapping. If $(x_{0},y_{0})\in O$, $G(x_{0},y_{0})=0$ and $h\mapsto D_yG(x_{0},y_{0})[h]$ is a Banach space isomorphism from $Y$ onto $Z$, then there exist neighbourhoods $U$ of $x_0$ and $V$ of $y_0$ and a continuously Fr\'echet differentiable function $f : U \to V$ such that $G(x, f(x)) = 0$, and $G(x, y) = 0$ if and only if $y = f(x)$, for all $(x,y)\in U\times V$.
\end{theorem}

We are now in position to conclude the proof of Proposition \ref{direct}: applying Theorem \ref{thimp} to $G(t,f)$ with $x_0=0$, $y_0=(0,0)$, $X=\Rm$, $Y=Z=W_r^{1,\infty} \times L^{\infty}_r$, it follows that there exists $t_0>0$ and $f \in C^1((-t_0,t_0), W_r^{1,\infty} \times L_r^\infty)$ (both $t_0$ and $f$ depend on $\varrho_1$ and $\varrho_2$) such that
$$
n[\varrho(t,f(t))]=n, \qquad k[\varrho(t,f(t))]=k.
$$

We then have that $\varrho(t,f(t)) \in \calA(n,k)$ when $\varrho(t,f(t))$ is positive, which is the case when $\varrho_2=0$ for all $t$, and when $t\geq 0$ when $\varrho_2 \neq 0$. Since moreover $f$ is continuously differentiable with values in $W^{1,\infty} \times L^\infty$, it is clear that $L(t,f(t)) \in C^1((-t_0,t_0), \calL(H^1,L^2))$, and as a consequence that $\partial_t \varrho(t,f(t))$ exists for $t \in (-t_0,t_0)$ and is bounded. To obtain that $\partial_t \varrho(t,f(t))$ is trace class, we proceed as in the proof of Lemma \ref{rhotE}, and show that $\partial_t L(t,f(t)) \sqrt{\vs}$ belongs to $\calJ_2$.

Finally, to obtain the system \fref{sysf} on $f_2'(0)$ and $f_3'(0)$, we differentiate the equation $G(t,f(t))=0$, and find at $t = 0$,
\begin{equation*}
\partial_t G(0,0) + D_fG(0,0)[f'(0)] = 0.
\end{equation*}
With
$$
\partial_t G_1(0,0)=n[\varrho_1 \vs+ \vs \varrho_1^*] +n[\varrho_2], \qquad \partial_t G_2(0,0)=k[\varrho_1 \vs+ \vs \varrho_1^*] +k[\varrho_2],
$$
and by following the steps of the proof of (ii) in Lemma \ref{lemiso} with $n_1+n[\varrho_2]$ and $k_1+k[\varrho_2]$ in place of $- z_1$ and $-z_2$, we recover the desired result and end the proof of Proposition \ref{direct} provided we prove Lemma \ref{lemiso}.

\paragraph{Proof of Lemma \ref{lemiso}.}
  The function $G$ is by construction a second order polynomial of $(t,f_2,f_3)$, and as such establishing continuity and differentiability is fairly direct for functions $f$ as regular as $f=(f_2,f_3) \in W^{1,\infty} \times L^\infty$. We will then simply prove for item (i) that $G \in W^{1,\infty} \times L^\infty$, and, denoting by $D_f G(t,f)[h]$ the differential of $G$ w.r.t $f$ at the point $(t,f)$ in the direction $h$, that $ \partial_t G(t,f), D_f G(t,f)[h] \in W^{1,\infty} \times L^\infty$ for any $h \in W^{1,\infty} \times L^\infty$, leaving the technical details about continuity and derivability out.

   Since $\vs$ and $\varrho_2$ have an identical role in \fref{decomp}, we set $\varrho_2=0$ until further notice without lack of generality. We start with the bounds on $G$. Lemma \ref{rhotE} shows that $\varrho(t,f) \in \calE$, and as a consequence $\sqrt{n[\varrho(t,f)]} \in H^1$ since $\varrho(t,f)$ is positive. This yields $n[\varrho(t,f)] \in L^\infty$ by a Sobolev embedding. Moreover, since the following pointwise estimate holds, see Remark \ref{rem2},
  $$|\nabla n[\varrho(t,f)|] \leq 2 \sqrt{n[\varrho(t,f)] k[\varrho(t,f)]},$$
  it suffices to show that $k[\varrho(t,f)] \in L^\infty$ to conclude that $G \in W^{1,\infty} \times L^\infty$. Since $\varrho(t,f) \in \calE$, Remark \ref{rem1} shows that
\begin{equation*}
  k[\varrho(t,f)] = \sum_{j \in \Nm} \rho_j |\nabla (1+L(t,f)) \phi_j|^2,
\end{equation*}
with convergence in $L^1$ and almost surely. From \fref{nabL}, we deduce, $x$ a.e.,
\begin{align}
|\nabla(1+L(t,f))\phi_j|&\leq  |\nabla \varrho_1 \phi_j| + \left(1+\|f_2\|_{L^{\infty}} + \|f_3\|_{L^{\infty}} \right) | \nabla \phi_j|+\|\nabla f_2\|_{L^\infty} |\phi_j|. \label{estnabL}
\end{align}
It is clear that  % Recalling that $\varrho_1$ is self-adjoint and trace class, we have, for all $\varphi \in L^2$,
% \begin{align*}
% \varrho_1\varphi&=  \sum_{j\in\mathbb{N}} \rho_j[\varrho_1] \langle \phi_j[\varrho_1] , \varphi \rangle  \phi_j[\varrho_1],
% \end{align*}
% with convergence e.g. in $L^2$ (with $\rho_j[\varrho_1]$ and $\phi_j[\varrho_1]$ the eigenvalues and eigenfunctions of $\varrho_1$). Since actually $\varrho_1 \in \calE$, the series above can be differentiated term by term with convergence in $L^2$. Since moreover $k[\varrho_1] \in L^\infty$, the differentiated series can be shown to be absolutely convergent in $L^\infty$ according to Remark \ref{rem1}. We then have, after using the Cauchy-Schwarz inequality,
\be \label{estnab1}
\|\nabla \varrho_1\varphi_j\|_{L^\infty} \lesssim \|\nabla \varphi\|_{L^\infty},
\ee
and going back to \fref{estnabL}, we find the estimate
\be \label{estL1}
 |\nabla(1+L(t,f))\phi_j| \lesssim (1+ |\phi_j|+ |\nabla \phi_j|), \qquad a.e.
\ee
This yields $k[\varrho(t,f)] \lesssim 1+ n+ k$ a.e., and therefore $k[\varrho(t,f)]$ is in $L^\infty$ since both $n$ and $k$ are bounded according to Assumptions A. This yields $G \in W^{1,\infty} \times L^\infty$.

We now consider the differential of $G$ with respect to $f$ in the direction $h=(h_2,h_3) \in W^{1,\infty} \times L^\infty$ and prove some estimates. Direct calculations lead  formally to
\begin{align*}
& D_f G_1(t,f)[h] = 2 \Re \sum_{j \in \Nm} \rho_j \left(( h_2 + T_{h_3}) \phi_j \right)^* (1+L(t,f)) \phi_j\\
&D_f G_2(t,f)[h] = 2 \Re \sum_{j \in \Nm} \rho_j \left( (\nabla h_2)\phi_j + (h_2+h_3) \nabla \phi_j\right)^* \nabla\big( (1+L(t,f)) \phi_j \big),
\end{align*}
where $\rho_j$ and $\phi_j$ the eigenvalues and eigenfunctions of $\vs$.  We show first that the series above converge almost everywhere. Since $|T_{h_3}\phi_j|\leq \|h_3\|_{L^{\infty}} \|\nabla \phi_j\|_{L^2}$, we deduce
\begin{align*}
&|(h_2 + T_{h_3})\phi_j|\leq \|h_2\|_{L^{\infty}} |\phi_j| + \|h_3\|_{L^{\infty}} \|\nabla \phi_j\|_{L^2}\\
&|\nabla (h_2 + T_{h_3})\phi_j| \leq \|\nabla h_2\|_{L^{\infty}} |\phi_j| + (\|h_2\|_{L^{\infty}} + \|h_3\|_{L^{\infty}}) |\nabla \phi_j|.
\end{align*}
The above inequalities, together with \fref{estL1}, show that  $D_f G_1(t,f)[h]$, $\nabla D_f G_1(t,f)[h]$ and $D_f G_2(t,f)[h]$ all converge a.e. according to Remark \ref{rem1}. We have moreover the estimate
$$
|D_f G_1(t,f)[h]|+|\nabla D_f G_1(t,f)[h]|+ |D_f G_2(t,f)[h]|\lesssim 1+n+k \qquad a.e.,
$$
 leading to $D_fG_1(t,f)[h]\in W^{1,\infty}$ and $D_fG_2(t,f)[h]\in L^{\infty}$ according to Assumptions A.

We now consider the time derivative of $G$, and do not assume anymore that $\varrho_2 = 0$. We have formally
\begin{align*} %\label{dG1}
&\partial_t G_1(t,f)=  2 \Re \sum_{j \in \Nm} \rho_j (\varrho_1 \phi_j)^* (1+L(t,f)) \phi_j + \sum_{j \in \Nm} \rho_j[\varrho_2] |(1+L(t,f)) \phi_j[\varrho_2]|^2\\ %\label{dG2}
&\partial_t G_2(t,f)= 2 \Re \sum_{j \in \Nm} \rho_j \big(\nabla \big(\varrho_1 \phi_j\big)\big)^* \nabla\big( (1+L(t,f)) \phi_j \big) + \sum_{j \in \Nm} \rho_j[\varrho_2] |\nabla (1+L(t,f)) \phi_j[\varrho_2]|^2,
\end{align*}
and need to prove that the series converge almost everywhere to functions bounded in $W^{1,\infty}$ and in $L^\infty$, respectively. These facts are easily established by following the same lines as above, using \fref{eL}-\fref{nabL}-\fref{estnab1}-\fref{estL1} and substituting $(\rho_j,\phi_j)$ by $(\rho_j[\varrho_2],\phi_j[\varrho_2])$ when necessary. As mentioned at the beginning of the proof, we do not give more details about the continuity of $G$ and its derivatives as the proofs are essentially identical to those of the bounds above. At that stage, we have therefore obtained (i), and we focus now on (ii).

%\medskip

\paragraph{Proof of (ii), step 1:} We first simplify the expression of the differential. Let $h=(h_2,h_3) \in W_r^{1,\infty} \times L_r^{\infty}$. Since $L(0,0) = 0$, we find
\begin{align*}
D_f G_1(0,0)[h] &= 2\Re \sum_{j\in\mathbb{N}} \rho_j\left((h_2 +  T_{h_3})\phi_j^* \phi_j\right)
\\ D_f G_2(0,0)[h] &=2 \Re \sum_{j \in \Nm} \rho_j \nabla \big (( h_2 + T_{h_3}) \phi_j^*\big) \nabla \phi_j.
\end{align*}
Furthermore, we can see that
\begin{equation*}
2\Re \sum_{j\in\mathbb{N}} \rho_j \left(h_2 +  T_{h_3}\right)\phi_j^* \phi_j= 2h_2n + \mathcal{L}_{K_0} h_3,
\end{equation*}
where $\calL_{K_0}$ is defined in \fref{defK}-\fref{defL}, and, since $\nabla T_{h_3}\phi_j^* = h_3 \nabla \phi_j^*$,
\begin{align*}
2 \Re \sum_{j \in \Nm} \rho_j \nabla \big (( h_2 + T_{h_3}) \phi^*_j\big) \nabla \phi_j & = 2 \Re \sum_{j \in \Nm} \rho_j ( (\nabla h_2) \phi^*_j+ \nabla T_{h_3}\phi^*_j)  \nabla \phi_j + \rho_j h_2 |\nabla \phi_j|^2
\\ & = 2 \Re \sum_{j \in \Nm} \rho_j  (\nabla h_2) \phi^*_j\nabla \phi_j  + 2  \rho_j (h_2+h_3) |\nabla \phi_j|^2
\\ &= \nabla n \nabla h_2 + 2(h_2+h_3) k,
\end{align*}
where we used Remark \ref{rem1} and \fref{nabn}. We now show that $D_f G(0,0)$ is invertible. 
\paragraph{Proof of (ii), Step 2:} Set $z=(z_1,z_2) \in W_r^{1,\infty} \times L_r^{\infty}$, and consider the equation 
%\begin{equation}\label{eq:DGimpl}
$$
D_fG(0,0)[h]=z,
$$
%\end{equation}
 which can be recast as, according to step 1,
\begin{equation} \label{system}\left\{\begin{array}{ll}
\displaystyle 2nh_2+ \calL_{K_0} h_3 = z_1,
\\\displaystyle \nabla n \nabla h_2  +  2k( h_2 + h_3)= z_2.
\end{array}\right.
\end{equation}
The above system can be reduced to a single equation on $h_3$. Indeed, the first equation of \fref{system} leads to
\begin{equation} \label{eqh2}
h_2 = -\frac{\mathcal{L}_{K_0}h_3}{2 n}+\frac{z_1}{2 n}.
\end{equation}
Differentiating, using that $K_0(x,x) = \nabla n(x)$ and $\nabla \mathcal{L}_{K_0} \varphi = (\nabla n) \varphi + \mathcal{L}_{\nabla_x K_0} \varphi$, we find
\begin{equation}\label{eq:nh2}
\nabla h_2 = -\frac{\nabla n}{2 n^2}\left( \mathcal{L}_{K_0} h_3 -  z_1 \right) + \frac{1}{2n}\left(\nabla z_1  - (\nabla n) h_3  - \mathcal{L}_{\nabla_x K_0} h_3\right).
\end{equation}
%Hence, we have
%\begin{equation*}
%2k h_2 = -\frac{k}{n}\mathcal{L}_{K_0}h_3+\frac{k z_1}{n},
%\end{equation*}
With $\frac{\nabla n}{2n} = \frac{\nabla \sqrt{n}}{\sqrt{n}}$ and $K_1(x,y)=-\frac{|\nabla n(x)|^2}{2 n^2(x)}K_0(x,y)-\frac{\nabla n(x)}{2n(x)}\nabla_x K_0(x,y) $, we deduce
\begin{align*}
\nabla n \nabla h_2 &= - \frac{|\nabla n|^2}{2n} \left(h_3 + \frac{ \mathcal{L}_{K_0} h_3}{n}\right) - \frac{\nabla n}{2n} \mathcal{L}_{\nabla_x K_0}h_3+ \nabla n  \nabla \left( \frac{z_1}{2n}\right)
\\ &= - \frac{|\nabla \sqrt{n}|^2}{2} h_3 -  \mathcal{L}_{K_1} h_3 +   \nabla n  \nabla \left( \frac{z_1}{2n}\right).
\end{align*}
By substituting the expression of $h_2$ into the second equation of \fref{system}, and recalling that \begin{equation*}
2k -\frac{|\nabla n|^2}{2n} = 2k - 2 |\nabla\sqrt{n}|^2 = 2a^{-1},
\end{equation*}
we finally obtain
\begin{align*}
\frac{k}{n}z_1 +\nabla n \nabla \left(\frac{z_1}{2 n}\right) - \mathcal{L}_{K_1}h_3  +\frac2ah_3=z_2,
\end{align*}
which can be recast as a Volterra equation
\be \label{eqh3}
h_3(x)=\mathcal{L}_{K}h_3(x)+  r(x), \quad x\in (0,1),
\ee
where $K$ is defined in \fref{defK} and
\begin{equation*}
r=\frac{a}{2}\left(z_2-\frac{k}{n} z_1 - \nabla \left(\frac{z_1}{2 n}\right)  \nabla n\right) =  - \frac{z_1}{2n} + \frac{a}{2}\left(z_2- \frac{\nabla\sqrt{n}}{\sqrt{n}}  \nabla z_1\right).
\end{equation*}
Since all $z_1, n^{-1},a,z_2, \nabla \sqrt{n}$ and $\nabla z_1$ are bounded, it follows that $r \in L^\infty$ %as
%\begin{align*}
%\|r\|_{L^{\infty}} &\leq \frac{a_M}2 \|z_2\|_{L^{\infty}} + n_m^{-1}\|z_1\|_{L^{\infty}} + \frac{a_Mn_m^{-1/2}}2 \|\nabla\sqrt{n}\|_{L^{\infty}} \|\nabla z_1\|_{L^{\infty}}<+\infty,
%\end{align*}
and that equation \fref{eqh3}  admits a unique solution in $L^{\infty}$ according to Lemma \ref{lem:Volterra} further since $K \in L^\infty \times L^\infty$ as mentioned below \fref{defK} (note that $\nabla_x K$ belongs to $L^\infty \times L^\infty$ as well). The solution is real since the coefficients in \fref{eqh3} are real. To conclude the proof of the lemma, it remains to verify that $h_2 \in W^{1,\infty}$.
\paragraph{Proof of (ii), conclusion:} It is clear from \fref{eqh2} that $h_2$ is real and bounded since $K_0, n^{-1},h_3$ and $z_1$ are bounded. Furthermore, we have just seen that $\nabla_x K_0$ is bounded, and so are $\nabla n$ and $\nabla z_1$. Equation \fref{eq:nh2} finally shows that $\nabla h_2 \in L^\infty$.
% \begin{equation*}
% %\|h_2\|_{L^{\infty}} \leq n_m^{-1} \left(\|K_0\|_{(L^{\infty})^2} \|h_3\|_{L^{\infty}} + 2^{-1}\|z_1\|_{L^{\infty}}\right),
% \end{equation*}
% and, by  \eqref{eq:nh2},
% \begin{align*}
% \|\nabla h_2\|_{L^{\infty}} \leq &\;  n_m^{-3/2}\|\nabla \sqrt{n}\|_{L^{\infty}}  \|K_0\|_{(L^{\infty})^2} \|h_3\|_{L^{\infty}}
% \\  &+  (2n_m)^{-1} \left(\|\nabla_x K_0\|_{(L^{\infty})^2} + 2 n_m^{1/2}\|\nabla \sqrt{n}\|_{L^{\infty}})\|h_3\|_{L^{\infty}} + \|\nabla z_1\|_{L^{\infty}}\right)
% \\ &+ n_m^{-3/2}\|\nabla \sqrt{n}\|_{L^{\infty}} \|z_1\|_{L^{\infty}}.
% \end{align*}
It follows that $h_2\in W_r^{1,\infty}$ and then that $DG(0,0)$ is an isomorphism from $W_r^{1,\infty}\times L_r^{\infty}$ to itself. 
This ends the proof. %By the implicit function theorem, we conclude that there exists a neighborhood $U\subset [-1,1]$ of $\{0\}$ and a neighborhood $V\in W^{1,\infty}\times L^{\infty}$ of $\{0\}$ such that there exists a function $f(t) : U \to V$ which is Fr\'echet differentiable and verifies
%\begin{equation}
%G(t,f(t)) = 0,\quad\forall t\in U,\quad\textrm{and}\quad f(0) = 0. \label{eq:ImplicitFunRes}
% \end{equation}
\bigskip

The next lemma provides us with the existence and uniqueness of solutions of integral equations of the form \fref{eqh3}.

\begin{lemma}\label{lem:Volterra}
Consider a kernel $N \in L^\infty \times L^\infty$.
Then, for any $\psi\in L^p$, $p\in(1,\infty]$, the Volterra equation
\begin{equation}\label{eq:Volterra}
\varphi(x) = \calL_{N} \varphi(x)+\psi(x)=\int_0^x N(x,y) \varphi(y)dy + \psi(x),\quad x\in (0,1),
\end{equation}
and its adjoint
\begin{equation}\label{eq:adjVolterra}
\varphi(x) = \calL^*_{N} \varphi(x)+\psi(x)=\int_x^1 N(y,x) \varphi(y)dy + \psi(x),\quad x\in (0,1),
\end{equation}
both admit a unique solution in $L^p$. Moreover, both solutions satisfy the estimate
\be \label{estVolt}
\| \varphi \|_{L^2} \leq e^{(\kappa+1)^2}(\kappa+1) \|\psi\|_{L^2}
\ee
where $\kappa$ is such that
$$
\|N\|_{L^\infty \times L^\infty} \leq \kappa.
$$
\end{lemma}

\begin{proof}
For $\varphi \in L^p$ and a given $\alpha>0$, consider the norm
$$
\|\varphi\|_{L^{p}_{\alpha}} = \|\varphi e^{-\alpha \cdot}\|_{L^p},
$$
which is equivalent to $\|\varphi\|_{L^p}$ since the domain of integration is bounded. With H\"older's inequality, we obtain that
\begin{align*}
\|\mathcal{L}_{N} \varphi\|_{L^p_{\alpha}}  &=\left\|  \int_0^{\cdot}N(\cdot,y)e^{-\alpha(\cdot-y)} e^{-\alpha y} \varphi(y)dy\right\|_{L^p}
\\ &\leq \kappa \sup_{x\in(0,1)} \left| \int_0^x e^{-\alpha(x-y) p'}dy \right|^{p'} \| \varphi\|_{L^{p}_{\alpha}} = \frac{\kappa}{(\alpha p')^{1/p'}} \left(1- e^{-\alpha p' }\right)^{1/p'}\|\varphi\|_{L^{p}_{\alpha}}
\\ &\leq \frac{\kappa}{(\alpha)^{1/p'}} \|\varphi\|_{L^{p}_{\alpha}}.
\end{align*}
Setting $\alpha = (\kappa+1)^{p'}$, it follows that $\calL_{N}$ is a contraction in $L^p$ equipped with the norm $\|\cdot\|_{L^{p}_{\alpha}}$. As a consequence, \fref{eq:Volterra} admits a unique solution in $L^p$, which satisfies the estimate 
\begin{equation*}
\|\varphi\|_{L^{p}_{(\kappa+1)^{p'}}} \leq \frac{\kappa}{\kappa+1}  \|\varphi\|_{L^{p}_{(\kappa+1)^{p'}}} + \|\psi\|_{L^{p}},
\end{equation*}
and \fref{estVolt} is obtained by setting $p=2$ and by remarking that $\|\varphi\|_{L^p_\alpha} \geq e^{-\alpha}\|\varphi\|_{L^p}$.  The proof for \eqref{eq:adjVolterra} follows from the same arguments by considering $L^{p}_{-(\kappa+1)^{p'}}$.
\end{proof}

\subsection{Proof of Lemma \ref{estimnlog}} \label{proofE1}
We first recall from \cite[Lemma A.1]{MP-JSP} the following estimate, for any $\varrho \in \calE^+$,
%\begin{equation}\label{eqLieb}
$$
\sum_{j \in \Nm} \lambda_j \rho_j \leq \Tr \big(\sqrt{H_0} \varrho \sqrt{H_0} \big),
$$
%\end{equation}
where $\{\lambda_j\}_{j \in \Nm}$ is the nondecreasing sequence of eigenvalues of $H_0$ with domain $H^2_{\rm{Neu}}$, and $\{\rho_j\}_{j \in \Nm}$ the eigenvalues of $\varrho$. This yields in particular, using H\"older's inequality, 
\be \label{ine23}
\sum_{j \geq N} (\rho_j)^{2/3} \leq \left(\sum_{j \geq N} \lambda_j \rho_j \right)^{2/3} \left(\sum_{j \geq N} (\lambda_j)^{-2} \right)^{1/3} \lesssim \|\varrho\|^{2/3}_{\calE},
\ee
where we used that $\lambda_j$ behaves like $j^2$. With $|\rho_j| \leq \|\varrho\|_{\calJ_1}$, we have then obtained the estimate, for all $\varrho \in \calE^+$,
\begin{equation}\label{eq:estm23}
\Tr(\varrho^{2/3}) \lesssim \|\varrho\|_{\mathcal{E}}^{2/3}.
\end{equation}
We are now in position to prove the Lemma. First, it is clear that $\varrho_{\star,\eta} \log(\varrho_{\star,\eta}+\eta)$ is trace class since $\varrho_{\star,\eta}$ is trace class and $\log(\varrho_{\star,\eta}+\eta)$ is bounded. Hence, $n[\varrho_{\star,\eta} \log(\varrho_{\star,\eta})]$ is well-defined in $L^1$ and
$$
n[\varrho_{\star,\eta} \log(\varrho_{\star,\eta}+\eta)]=\sum_{j \in \Nm} \rho_{j,\eta} \log(\rho_{j,\eta} + \eta) |\phi_{j,\eta}|^2,
$$
with convergence in $L^1$ and a.e. We will show that the series actually converges in $L^\infty$. Above, $\rho_{j,\eta}$ and  $\phi_{j,\eta}$ are the eigenvalues and eigenvectors of $\varrho_{\star,\eta}$. With the triangle inequality and the inequality $\|u\|^2_{L^\infty} \lesssim \|u\|_{L^2} \|u\|_{H^1}$ for $u \in H^1$, we find
    \begin{align*}
\|n[\varrho_{\star,\eta} \log(\varrho_{\star,\eta}+\eta)]\|_{L^\infty} &\leq \sum_{ j \in \Nm} \rho_{j,\eta}  |\log(\rho_{j,\eta}+\eta)| \|\phi_{j,\eta}\|_{L^\infty}^2\\
%\\ &\leq \sum_{ j \in \Nm} \rho_{j,\eta}  |\log(\rho_{j,\eta} + \eta)| \|\phi_{j,\eta}\|_{L^{\infty}} \|\phi_{j,\eta}\|_{L^2}
 &\lesssim \left(\sum_{ j \in \Nm} \rho_{j,\eta}  |\log(\rho_{j,\eta} + \eta)|^2\right)^{1/2}\left(\sum_{ j \in \Nm} \rho_{j,\eta} \|\phi_{j,\eta}\|_{H^1}^2 \right)^{1/2}.
    \end{align*}
    The second term on the r.h.s is the square root of $\|\varrho_{\star,\eta}\|_{\calE}$, and is uniformly bounded in $\eta$ since by construction
    \be \label{UE}
    \|\varrho_{\star,\eta}\|_{\calE}=\|n\|_{L^1}+\|k\|_{L^1}.
    \ee
    With $x|\log(x)|^2 \leq C x^{2/3}$ for $x\in [0,1]$, where $C = \max_{x\in[0,1]} x^{1/3}|\log(x)|^2$, we control the first term of the r.h.s by the square root of $\Tr(\varrho_{\star,\eta}^{2/3})$, which is also uniformly bounded in $\eta$ thanks to \fref{eq:estm23} and \fref{UE}. This ends the proof.
 \subsection{Proof of Lemma \ref{lem:EulerLag1}} \label{proofEL1}
First of all, it follows from Proposition \ref{direct} that $\varrho_{\eta}(t,f(t))$ is differentiable with respect to $t$ with values in $\calJ_1$, and a direct calculation shows that
\begin{equation} \label{derivrho}
\left.\frac{d \varrho_{\eta}(t,f(t))}{dt}\right|_{t=0}= \left(\varrho_{1,\eta}+f_2'(0)I+T_{f_3'(0)} \right) \varrho_{\star,\eta}+\varrho_{\star,\eta}\left(\varrho^*_{1,\eta}+f_2'(0)I+T^*_{f_3'(0)}\right).
\end{equation}
Note that the last term above has to be interpreted as
$$
\sqrt{\varrho_{\star,\eta}}\, \overline{\sqrt{\varrho_{\star,\eta}}\left(\varrho^*_{1,\eta}+f_2'(0)I+T^*_{f_3'(0)}\right)},
$$
where both operators are in $\calJ_2$. We have in particular, $\forall \varphi\in L^2$, with convergence in $L^2$,
\begin{align*}
 \overline{\sqrt{\varrho_{\star,\eta}}T^*_{f_3'(0)}}\varphi &=\sum_{ j \in \Nm} \sqrt{\rho_{j,\eta}} \phi_{j,\eta} \langle T_{f_3'(0)} \phi_{j,\eta},  \varphi \rangle
\\ &=\sum_{ j \in \Nm}  \sqrt{\rho_{j,\eta}}\phi_{j,\eta} \int_0^1 \left(\int_0^x f_3'(0)(y) \nabla \phi_{j,\eta}^*(y)dy\right) \varphi(x) dx
\\ &=\sum_{ j \in \Nm}  \sqrt{\rho_{j,\eta}} \phi_{j,\eta} \int_0^1 \left(\int_y^1\varphi(x)  dx\right) f_3'(0)(y) \nabla \phi_{j,\eta}^*(y) dy,
\end{align*}
where $\rho_{j,\eta}$ and  $\phi_{j,\eta}$ are the eigenvalues and eigenvectors of $\varrho_{\star,\eta}$. We proceed now to the proof of the Lemma: the fact that $F_\eta \in C^1((-t_0,t_0))$ is a consequence of $\varrho_{\eta}(t,f(t)) \in C^1((-t_0,t_0),\calJ_1)$ that is obtained in Proposition \ref{direct}, and of an adaptation of Lemma 5.3 in   \cite{MP-JSP}: it is shown there that $S_\eta(\varrho)$ for $\varrho \in \calE^+$ is G\^ateaux differentiable in any direction of the form $\varrho+ t \omega$, where $\omega \in \calJ_1$, and the differential is $D S_\eta(\varrho)(\omega)=\Tr (\log (\varrho+\eta) \omega)$. It suffices for our purpose to generalize the proof to the more general form of perturbations $\varrho_\eta(t,f(t))$ that we consider here, and we then find, with \fref{derivrho},
\begin{align*}
\left.\frac{d F_{\eta}(t)}{dt}\right|_{t=0}=&\;\Tr\left( \log(\varrho_{\star,\eta} + \eta)\left(\varrho_{1,\eta}+f_2'(0)I+T_{f_3'(0)} \right) \varrho_{\star,\eta} \right)
\\&+\Tr\left(\log(\varrho_{\star,\eta}+\eta)\varrho_{\star,\eta}\left(\varrho^*_{1,\eta}+f_2'(0)I+T^*_{f_3'(0)}\right)\right).
\end{align*}
The traces above are well-defined since $\log(\varrho_{\star,\eta} + \eta)$ is bounded and $(\varrho_{1,\eta}+f_2'(0)+T_{f_3'(0)}) \varrho_{\star,\eta}$ and $\varrho_{\star,\eta}(\varrho^*_1+f_2'(0)+T^*_{f_3'(0)})$ are trace class. We now make the r.h.s above more explicit and only dependent on $f_3$, and treat first the terms involving $f_3$. We find
\begin{align} \nonumber
  \Tr \left(T_{f_3'(0)}  \varrho_{\star,\eta}  \log(\varrho_{\star,\eta} + \eta) \right) +&\Tr \left(\log(\varrho_{\star,\eta} + \eta) \varrho_{\star,\eta} T^*_{f_3'(0)}\right)
\\&=2 \Re \sum_{ j \in \Nm} \rho_{j,\eta} \log (\rho_{j,\eta} + \eta) \langle T_{f_3'(0)} \phi_{j,\eta},  \phi_{j,\eta}\rangle\nonumber
\\&= \langle c_{1,\eta},f_3'(0)\rangle, \label{eq:relTf3p}
\end{align}
where
\begin{equation*}
c_{1,\eta}(x)= 2 \Re \sum_{ j \in \Nm} \rho_{j,\eta} \log(\rho_{j,\eta} + \eta) \nabla \phi_{j,\eta}(x) \int_x^1 \phi_{j,\eta}^* (y) dy .
\end{equation*}
The series above converges almost everywhere to a bounded function since the partial sum is controlled by
$$
\max_{x \in [0,1]} |\log(x+\eta)|\left(\sum_{ j \leq N} \rho_{j,\eta} |\nabla \phi_{j,\eta}|^2\right)^{1/2} (\Tr (\varrho_{\star,\eta}))^{1/2},
$$
where the sum converges a.e. according to Remark \ref{rem1}, and the estimate $c_{1,\eta} \lesssim \sqrt{k}$ a.e. yields the bound according to Assumptions A. For the term in $f_2$, we use the expression of $f_2'(0)$ given in \fref{sysf} to obtain that
\begin{align}
\Tr( f_2'(0) & \varrho_{\star,\eta} \log(\varrho_{\star,\eta} + \eta) +  \log(\varrho_{\star,\eta} + \eta)\varrho_{\star,\eta} f_2'(0)) 
=  2 \langle f_2'(0),n[\varrho_{\star,\eta} \log(\varrho_{\star,\eta} + \eta)]\rangle \nonumber
\\[3mm] &= -  \left\langle \mathcal{L}_{K_{0,\eta}} f_3'(0), \frac{n[\varrho_{\star,\eta} \log(\varrho_{\star,\eta} + \eta)]}{n}\right\rangle - \left\langle n_1, \frac{n[\varrho_{\star,\eta} \log(\varrho_{\star,\eta} + \eta)]}{n}\right\rangle, \label{eq:relTf2p}
\end{align}
where $n_{1,\eta}=n[ \varrho_{1,\eta} \varrho_{\star,\eta}+\varrho_{\star,\eta} \varrho_{1,\eta}^*]$ and $\calL_{K_{0,\eta}}$ is as in \fref{defK}-\fref{defL} with $\rho_{j,\eta}$ and  $\phi_{j,\eta}$ in place of the eigenvalues and eigenvectors of $\vs$.
Recalling that, $\forall \varphi\in L^1$,
%\begin{equation}\label{eq:Ladjoin}
$$
\mathcal{L}_{K_{0,\eta}}^* \varphi (y) = \int_y^1 K_{0,\eta}(x,y) \varphi(x) dx, \qquad K_{0,\eta} \in L^\infty \times L^\infty,
$$
% \end{equation}
we deduce that
\begin{equation} \label{eqc2}
 \left\langle \mathcal{L}_{K_{0,\eta}} f_3'(0), \frac{n[\varrho_{\star,\eta} \log(\varrho_{\star,\eta} + \eta)]}{n}\right\rangle = \langle c_{2,\eta} , f_3'(0)\rangle,
\end{equation}
with
\begin{align*}
c_{2,\eta}(x) &=  - 2 \Re\sum_{ j \in \Nm} \rho_{j,\eta} \nabla \phi_{j,\eta}(x) \int_x^1 \phi_{j,\eta}^*(y)n[\varrho_{\star,\eta} \log(\varrho_{\star,\eta} + \eta)](y)(n(y))^{-1} dy.
\end{align*}
Following the same lines as $c_{1,\eta}$, it is direct to show that the series defining $c_{2,\eta}$ converges in $L^\infty$. The proof is ended by collecting \fref{eq:relTf3p}-\fref{eq:relTf2p}-\fref{eqc2}.
 \subsection{Proof of Lemma \ref{lem:conveta}} \label{proofconveta}
We only sketch the proof since most of the results can be found elsewhere. We prove first that $\varrho_{\star,\eta_\ell}$ converges to $\vs$ in $\calJ_1$. We know for this that $\varrho_{\star,\eta_\ell}$ is bounded in $\calE^+$ independently of $\ell$ since by construction
$$
\|\varrho_{\star,\eta_\ell}\|_{\calE}=\|n\|_{L^1}+\|k\|_{L^1}.
$$
Then, according to \cite[Lemma 3.1]{MP-JSP}, there exists $\sigma \in \calE^+$ and a subsequence (still denoted $\eta_\ell$) such that $\varrho_{\star,\eta_\ell} \to \sigma$ in $\calJ_1$ and $\sqrt{H_0} \varrho_{\star,\eta_\ell} \to \sqrt{H_0} \sigma$ in $\calJ_2$.  We need to prove that $\sigma=\vs$. For this, we have, since $\vs \in \calA(n,k)$,
$$
S_{\eta_\ell}(\varrho_{\star,\eta_\ell})\leq S_{\eta_\ell}(\vs).
$$
With \cite[Lemma 5.2 (iv)]{MP-JSP}, we can pass to the limit in the equation above to obtain
\be \label{ineqS}
S(\sigma)\leq S(\vs).
\ee
We are done if we can show that $\sigma \in \calA(n,k)$ since $S$ admits a unique minimizer. This is done by following the steps of the proof of Theorem \ref{thexist} obtained in \cite{DP-JMPA}. It is shown therein that the fact that $\sigma \in \calA(n,k)$ is established by comparing $\|k[\sigma]\|_{L^1}$ and $\|k\|_{L^1}$: based on \fref{ineqS}, it is proved that it is only possible that $\|k[\sigma]\|_{L^1}=\|k\|_{L^1}$, which is shown in \cite{DP-JMPA} to lead to the strong convergence of $\varrho_{\star,\eta_\ell}$ in $\calE$, and as a consequence to the fact that $k[\sigma]=k$.  We refer the reader to \cite{DP-JMPA} for more details. Since the minimizer is unique, the entire sequence converges to $\vs$. Once the convergence of $\varrho_{\star,\eta_\ell}$ to $\vs$ is established, items (1) and (2) follow from \cite[Lemma 3.1]{MP-JSP}. Items (3) and (4) are proved in \cite[Lemma A.2]{MP-JSP} and \cite[Lemma 3.7]{DP-JFA}. Item (5) is proved in  \cite[Lemma 5.2 (iv)]{MP-JSP}. Regarding (6), we write, for any $\varphi \in L^2$,
\be \label{normL}
\|\sqrt{\varrho_{\star,\eta_\ell}} \log(\varrho_{\star,\eta_\ell}+\eta_\ell)\varphi\|_{L^2}^2=\sum_{j \in \Nm} \rho_{j,\eta_\ell} (\log(\rho_{j,\eta_\ell}+\eta_\ell))^2 |\langle \phi_{j,\eta_\ell},\varphi \rangle |^2.
\ee
We split the sum for $j <N$ and $j>N$. The large $j$ part is treated by using \fref{ine23} as in the proof of Lemma \ref{estimnlog}, and is controlled by
$$
R(N) \|\varrho_{\star,\eta_\ell}\|_{\calE}^{2/3} \lesssim R(N),
$$
where $R$ is a function independent of $\ell$ which tends to $0$ as $N \to \infty$. With this, it is sufficient to consider a finite number of terms in \fref{normL}, and we use items (3)-(4) to pass to the limit. This establishes the convergence of the $L^2$ norm of $\sqrt{\varrho_{\star,\eta_\ell}} \log(\varrho_{\star,\eta_\ell}+\eta_\ell)\varphi$ to that of $\sqrt{\varrho_{\star}} \log(\varrho_{\star})\varphi$. The weak converge in $L^2$ follows in the same manner, and since weak convergence combined with convergence of the norm implies strong convergence, item (6) follows. This ends the proof.

 \bibliographystyle{plain}
\bibliography{bibliography.bib}

\begin{thebibliography}{10}

\bibitem{QDD-JCP}
P.~Degond, S.~Gallego, and F.~M{\'e}hats.
\newblock An entropic quantum drift-diffusion model for electron transport in
  resonant tunneling diodes.
\newblock {\em J. Comput. Phys.}, 221(1):226--249, 2007.

\bibitem{isotherme}
P.~Degond, S.~Gallego, and F.~M{\'e}hats.
\newblock Isothermal quantum hydrodynamics: derivation, asymptotic analysis,
  and simulation.
\newblock {\em Multiscale Model. Simul.}, 6(1):246--272, 2007.

\bibitem{QHD-CMS}
P.~Degond, S.~Gallego, and F.~M{\'e}hats.
\newblock On quantum hydrodynamic and quantum energy transport models.
\newblock {\em Commun. Math. Sci.}, 5(4):887--908, 2007.

\bibitem{DGMRlivre}
P.~Degond, S.~Gallego, F.~M{\'e}hats, and C.~Ringhofer.
\newblock Quantum hydrodynamic and diffusion models derived from the entropy
  principle.
\newblock In {\em Quantum transport}, volume 1946 of {\em Lecture Notes in
  Math.}, pages 111--168. Springer, Berlin, 2008.

\bibitem{DR}
P.~Degond and C.~Ringhofer.
\newblock Quantum moment hydrodynamics and the entropy principle.
\newblock {\em J. Statist. Phys.}, 112(3-4):587--628, 2003.

\bibitem{Dolbeault-Loss}
J.~Dolbeault, P.~Felmer, M.~Loss, and E.~Paturel.
\newblock Lieb-{T}hirring type inequalities and {G}agliardo-{N}irenberg
  inequalities for systems.
\newblock {\em J. Funct. Anal.}, 238(1):193--220, 2006.

\bibitem{DP-JFA}
R.~Duboscq and O.~Pinaud.
\newblock A constrained optimization problem in quantum statistical physics.
\newblock {\em Submitted}, 2019.

\bibitem{DP-JMPA}
R.~Duboscq and O.~Pinaud.
\newblock On the minimization of quantum entropies under local constraints.
\newblock {\em Journal de Math\'ematiques Pures et Appliqu\'ees}, 128:87--118,
  2019.

\bibitem{jungel-matthes}
A.~J{\"u}ngel and D.~Matthes.
\newblock A derivation of the isothermal quantum hydrodynamic equations using
  entropy minimization.
\newblock {\em ZAMM Z. Angew. Math. Mech.}, 85(11):806--814, 2005.

\bibitem{jungel-matthes-milisic}
A.~J{\"u}ngel, D.~Matthes, and J.~P. Mili{\v{s}}i{\'c}.
\newblock Derivation of new quantum hydrodynamic equations using entropy
  minimization.
\newblock {\em SIAM J. Appl. Math.}, 67(1):46--68, 2006.

\bibitem{junk}
M.~Junk.
\newblock Domain of definition of {L}evermore's five-moment system.
\newblock {\em J. Statist. Phys.}, 93(5-6):1143--1167, 1998.

\bibitem{levermore}
C.~D. Levermore.
\newblock Moment closure hierarchies for kinetic theories.
\newblock {\em J. Statist. Phys.}, 83(5-6):1021--1065, 1996.

\bibitem{LP}
P.-L. Lions and T.~Paul.
\newblock {Sur les mesures de Wigner}.
\newblock {\em Rev. Mat. Iberoamericana}, 9:553--618, 1993.

\bibitem{MP-JSP}
F.~M{\'e}hats and O.~Pinaud.
\newblock An inverse problem in quantum statistical physics.
\newblock {\em J. Stat. Phys.}, 140(3):565--602, 2010.

\bibitem{MP-KRM}
F.~M{\'e}hats and O.~Pinaud.
\newblock A problem of moment realizability in quantum statistical physics.
\newblock {\em Kinet. Relat. Models}, 4(4):1143--1158, 2011.

\bibitem{MP-JDE}
F.~M{\'e}hats and O.~Pinaud.
\newblock The quantum {L}iouville-{BGK} equation and the moment problem.
\newblock {\em J. of. Diff. Eq.}, 263(7):3737--3787, 2017.

\bibitem{nachter}
B.~Nachtergaele and H-T. Yau.
\newblock Derivation of the {E}uler equations from quantum dynamics.
\newblock {\em Comm. Math. Phys.}, 243(3):485--540, 2003.

\bibitem{pinaudPoincare}
O.~Pinaud.
\newblock The quantum drift-diffusion model: existence and exponential
  convergence to the equilibrium.
\newblock {\em Annales de l'Institut Henri Poincar\'e C, Analyse non
  lin\'eaire}, 36(3):811--836, 2019.

\bibitem{RS-80-I}
M.~Reed and B.~Simon.
\newblock {\em Methods of modern mathematical physics. I. Functional analysis}.
\newblock Academic Press, Inc., New York, second edition, 1980.

\end{thebibliography}

 \end{document}